\documentclass[12pt,onecolumn,draftclsnofoot]{IEEEtran}
\usepackage{graphicx}
\usepackage{float}
\usepackage[font={scriptsize}]{subcaption}
\usepackage[T1]{fontenc}
\usepackage{latexsym}
\usepackage[cmex10]{amsmath}
\usepackage[noadjust]{cite}
\usepackage{caption}
\usepackage{amssymb}
\usepackage{amsmath}
\usepackage{paralist}
\usepackage{leftidx}
\usepackage{mathrsfs}
\usepackage{multirow}
\usepackage{slashbox}
\usepackage{amsthm}
\usepackage{background}
\usepackage{bbm}
\usepackage{algorithm}
\usepackage{algorithmic}
\usepackage[norule]{footmisc}
\IEEEoverridecommandlockouts
\theoremstyle{definition}
\newtheorem{defn}{Definition}
\theoremstyle{proposition}
\newtheorem{prop}{Proposition}
\theoremstyle{lemma}
\newtheorem{lem}{Lemma}

\begin{document}
\SetBgContents{}
\setlength{\belowcaptionskip}{0pt}
\setlength{\textfloatsep}{15pt}
\title{Mobility-Aware Caching in D2D Networks}
\author{\IEEEauthorblockN{Rui Wang, Jun Zhang, S.H. Song and Khaled B. Letaief, \emph{Fellow, IEEE} }
\thanks{Rui Wang, Jun Zhang, S.H. Song and Khaled B. Letaief are with the Department of Electronic and Computer Engineering,
the Hong Kong University of Science and Technology, Clear Water Bay,
Kowloon, Hong Kong (e-mail: \{rwangae, eejzhang, eeshsong, eekhaled\}@ust.hk).
Khaled B. Letaief is also with Hamad bin Khalifa University, Doha,
Qatar (e-mail: kletaief@hkbu.edu.qa).}
\thanks{This work was supported by the Hong Kong Research Grant Council under Grant No. 610113.}}
\maketitle
\begin{abstract}
Caching at mobile devices can facilitate device-to-device (D2D) communications, which may significantly improve spectrum efficiency and alleviate the heavy burden on backhaul links. However, most previous works ignored user mobility, thus having limited practical applications. In this paper, we take advantage of the user mobility pattern by the inter-contact times between different users, and propose a mobility-aware caching placement strategy to maximize the \emph{data offloading ratio}, which is defined as the percentage of the requested data that can be delivered via D2D links rather than through base stations (BSs). Given the NP-hard caching placement problem, we first propose an optimal dynamic programming (DP) algorithm to obtain a performance benchmark, but with much lower complexity than exhaustive search. We will then prove that the problem falls in the category of monotone submodular maximization over a matroid constraint, and propose a time-efficient greedy algorithm, which achieves an approximation ratio as $\frac{1}{2}$. Simulation results with real-life data-sets will validate the effectiveness of our proposed mobility-aware caching placement strategy. We  observe that users moving at either a very low or very high speed should cache the most popular files, while users moving at a medium speed should cache less popular files to avoid duplication.
\end{abstract}
\IEEEpeerreviewmaketitle

\newpage

\section{Introduction}

With the popularity of smart phones, the data traffic generated by mobile applications, e.g., multimedia file sharing and video streaming, is undergoing an exponential growth, and will soon reach the capacity limit of current cellular networks. To meet the heavy demand, network densification is commonly adopted to achieve higher network capacity, which is expected to increase by 1000 times in future 5G networks \cite{5G,shi2014group,magaziney}. However, the increase of access points and user traffic put a heavy burden on backhaul links, which connect base stations (BSs) with the core network \cite{backhaul}. To reduce the burden, one promising approach is to cache popular contents at BSs and user devices \cite{femtocachingmaga,wang2014cache,cachebenefit,d2d-cache,cachsize,maddah2015decentralized}, so that mobile users can get the required content from local BSs or close user devices without utilizing backhaul links. Such local access can also reduce the download delay and improve the energy efficiency \cite{BScache,wang2014cache}.

Most previous investigations on wireless caching networks assumed fixed network topologies \cite{femtocaching2,femtocachinginfocom,femtocoded,femtocachingmag,d2d-cache,scaling,tradeoff,tradeoffj,fundamentallimits}. However, user mobility is an intrinsic feature of wireless networks, which changes the network topologies over time. Thus, it is critical to take the user mobility pattern into account. On the other hand, user mobility can also be a useful feature to exploit, as it will increase the communication opportunities of moving users. Mobility-aware design has been proved to be an effective approach to deal with lots of problems in wireless networks. For example, exploiting user mobility helps improve capacity in ad hoc networks \cite{mobilitycapacity} and reduce the probability of failed file delivery in femto-caching networks \cite{femtomobility}. In this paper, we will propose an effective mobility-aware caching strategy in device-to-device (D2D) caching networks to offload traffic from the cellular network.

\subsection{Related Works}

Compared with caching at BSs \cite{femtomobility,femtocaching2,femtocachinginfocom,BScache,femtocoded,jcache,li2015distributed,liu2015exploiting}, caching at mobile devices provides new features and unique advantages. First, the aggregate caching capacity grows with the number of devices in the network, and thus the benefit of device caching improves as the network size increases \cite{d2d-cache}. Second, by caching popular content at devices, mobile users may acquire required files from close user devices via D2D communications, rather than through the BS \cite{D2D,blue}. This will significantly reduce the mobile traffic on the backbone network and alleviate the heavy burden on the backhaul links. However, mobile caching also faces some new challenges. Specifically, not only the users, but also the caching helpers are moving over time, which brings additional difficulties in the caching design.

There have been lots of efforts on D2D caching networks while assuming fixed network topologies. Considering a single cell scenario, it was shown in \cite{d2d-cache} that D2D caching outperforms other schemes, including the conventional unicasting scheme, the harmonic broadcasting scheme, and the coded multicast scheme. Assuming that each device can store one file, Golrezaei \emph{et al.} analyzed the scaling behavior of the number of active links in a D2D caching network as the number of mobile devices increases \cite{scaling}. It was found that the concentration of the file request distribution affects the scaling laws, and three concentration regimes were identified. In \cite{tradeoff,tradeoffj}, the outage-throughput tradeoff in D2D caching networks was investigated and optimal scaling laws of per-user throughput were derived as the numbers of mobile devices and files in the library grow under a simple uncoded protocol. Meanwhile, the case using the coded delivery scheme \cite{fundamental} was investigated in \cite{fundamentallimits}.

There are some preliminary studies considering user mobility. In \cite{femtomobility}, Poularakis \emph{et al.} studied a femto-caching network with user mobility. A Markov chain model was adopted to represent which helper, i.e., a particular femtocell BS, was accessed by a specific user in different time slots. However, in D2D caching networks, such a model cannot be adopted, since there is no fixed caching helper, and all the mobile users may move over time. The effect of user mobility on D2D caching was investigated via simulations in \cite{D2Dmobility}, which showed that user mobility does not have a significant impact on a random caching scheme. However, such a caching scheme failed to take advantage of the user mobility pattern.  In \cite{krishnan2016effect}, it showed that user mobility has positive effect on D2D caching. In \cite{mobilityD2Dcaching}, Lan \emph{et al.} considered the case where mobile users can update caching content based on the file requirement and user mobility. However, it was assumed that one complete file can be transmitted via any D2D link when two users contact, which is not practical, considering the limited communication time and transmission rate. More recently, several design methodologies for mobility-aware caching were proposed in \cite{magmobility}, but more thorough investigations will be needed, especially on practical implementations.

\subsection{Contributions}

In this paper, we propose a mobility-aware caching placement strategy in D2D networks. The information of user connectivity in the user mobility pattern is captured with the inter-contact model \cite{exintercontactmodel,magazine}. Specifically, as shown in Fig. \ref{intercontact}, for an arbitrary pair of mobile users, the timeline consists of both \emph{contact times}, defined as the times when the users are within the transmission range, and \emph{inter-contact times}, defined as the times between contact times \cite{pocket}.
Each file is encoded into several segments by the rateless Fountain code \cite{Fountaincodes}, and a mobile user needs to collect enough encoded segments to recover the requested file. Once a user is in contact with some other users who cache part of its requested file, it will get some of the segments via D2D communication. Our objective is to design caching placement to maximize the data offloading ratio. The main contributions of this paper are summarized as follows:
\begin{itemize}
\item As each mobile user may get the requested file via multiple contacts with other users, the complexity of calculating the objective function, i.e., the data offloading ratio, increases exponentially with the number of users, which brings a major difficulty for algorithm design. We first propose a divide and conquer algorithm to efficiently evaluate the objective, with quadratic complexity with respect to the number of users.
\item The caching placement problem is shown to be NP-hard, and a dynamic programming (DP) algorithm is proposed to obtain the optimal solution with much lower complexity compared to exhaustive search. Although the DP algorithm is still impractical, it can serve as a performance benchmark for systems with small to medium sizes. Moreover, its computation complexity is linear with the number of files, and thus it can easily deal with a large file library.
\item To propose a practical solution, we reformulate the problem and prove that it is a monotone submodular maximization problem over a matroid constraint. The main contribution in this part is to prove that the complicated objective function is a monotone submodular function. With the reformulated form, a greedy algorithm is developed, which achieves at least $\frac{1}{2}$ of the optimal value.
\item Simulation results show that the proposed mobility-aware caching strategy provides significant performance gains over the commonly used popular caching strategy \cite{pupolarcaching} and random caching strategy \cite{randomcache}. Meanwhile, the greedy algorithm achieves near-optimal performance. Moreover, we observe that users with very low velocity tend to rarely contact with others, and thus they need to cache the most popular files to meet their own requirements. On the other hand, the users with high velocity contact others frequently, and they need to cache the most popular files as well, in order to help others to download the requested files. Meanwhile, the users with medium velocity need to cache less popular files to avoid duplicate caching.
\end{itemize}

\subsection{Organization}

The remainder of this paper is organized as follows. In Section II, we introduce the system model and formulate the mobility-aware caching placement problem. An optimal DP algorithm is proposed to serve as the performance benchmark in Section III. In Section IV, we reformulate the caching placement problem as a monotone submodular maximization problem over a matroid constraint, and adopt a greedy algorithm to solve it. The simulation results are provided in Section IV. Finally, we conclude this paper in Section V.


\section{System Model and Problem Formulation}

\begin{figure}
\begin{subfigure}{0.55\textwidth}
\includegraphics[width=\linewidth]{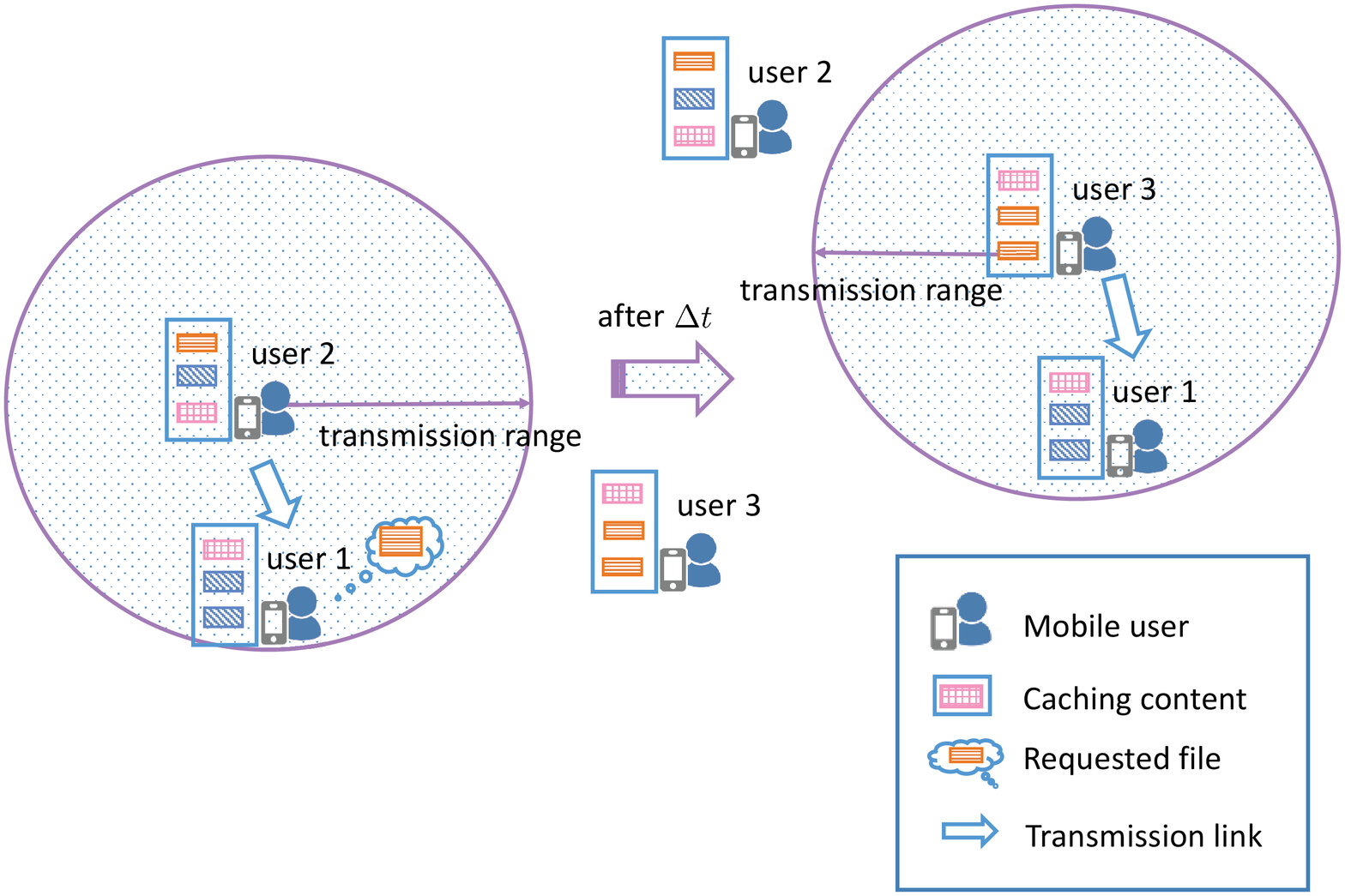}
\caption{A sample network.} \label{fig:1a}
\end{subfigure}
\hspace*{\fill} 
\begin{subfigure}{0.35\textwidth}
\vspace*{20mm}
\includegraphics[width=\linewidth]{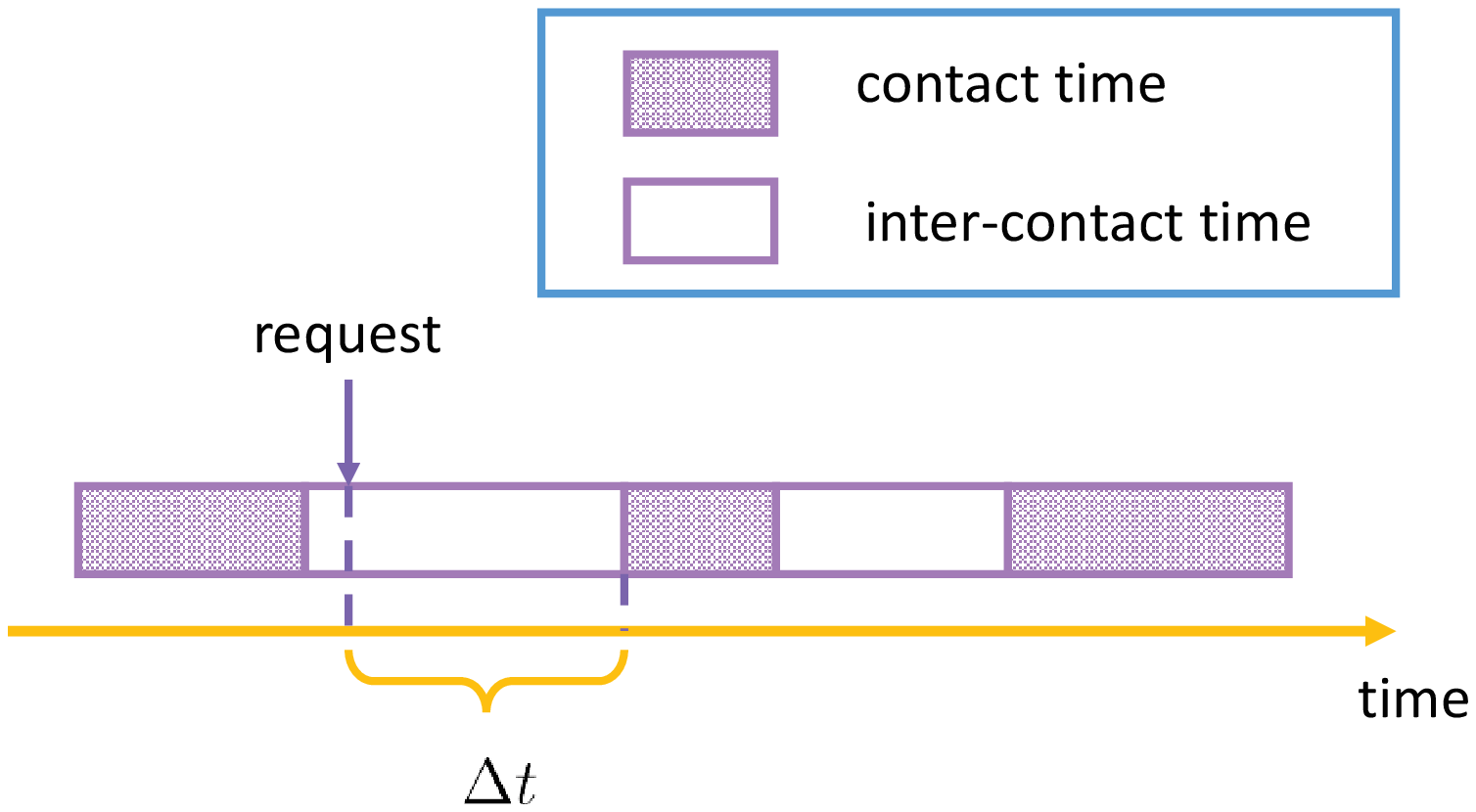}
\caption{The timeline of user 1 and user 3.} \label{intercontact}
\end{subfigure}
\caption{A sample network with $N_u=3$, $N_{file}=3$, $C=3$, $B=1$, and $K_f=2$ for $f \in \mathcal{F}$.} \label{model}
\end{figure}

In this section, we will firstly introduce the mobility model adopted in our study, together with the caching strategy and file transmission model. The caching placement problem will then be formulated.

\subsection{User Mobility Model}
The inter-contact model can capture the connectivity information in the user mobility pattern, and has been widely investigated in wireless networks \cite{exintercontactmodel,agginter,intercontactmodel,aggregate_real}. Thus, it is adopted in this paper to model the mobility pattern of mobile users. In this model, mobile users may contact with each other when they are within the transmission range. Correspondingly, the \emph{contact time} for two mobile users is defined as the time that they can contact with each other, i.e., they may exchange files during the contact time. Then, the \emph{inter-contact time} for two mobile devices is defined as the time between two consecutive contact times. Specifically, we consider a network with $N_u$ mobile users, whose index set is denoted as $\mathcal{D}=\{1,2,...,N_u\}$. Same as \cite{exintercontactmodel}, we model the locations of contact times in the timeline for any two users $i$ and $j$ as a Poisson process with intensity $\lambda_{i,j}$. For simplification, we assume that the timelines for different device pairs are independent. We call $\lambda_{i,j}$ as the \emph{pairwise contact rate} between users $i$ and $j$, which represents the average number of contacts per unit time. The pairwise contact rate can be estimated from historical data, and we will test our results on real-life data-sets in the simulation.

\subsection{Caching Strategy and File Delivery Model}
We consider a library of $N_{file}$ files, whose index set is denoted as $\mathcal{F}=\{1,2,...,N_{file}\}$. Rateless Fountain coding is applied \cite{Fountaincodes,femtomobility}, where each file is encoded into a large number of different segments, and it can be recovered by collecting a certain number of encoded segments. Accordingly, it can be guaranteed that there is no repetitive encoded segment in the network. Specifically, we assume that each file is encoded into multiple segments, each with size $s$ bits, and file $f$ can be recovered by collecting $K_f$ encoded segments. Note that the value of $K_f$ depends on the size of file $f$. It is assumed that each mobile user reserves a certain amount of storage capacity for caching, which can store at most $C$ encoded segments. The number of encoded segments of file $f$ cached in user $i$ is denoted as $x_{i,f}$.

Mobile users will request files in the file library based on their demands. For simplicity, we assume that the requests of all the users follow the same distribution, and file $f$ is requested by one user with probability $p_f$, where $\sum \limits_{f \in \mathcal{F}} p_f =1$. When a user $i$ requests a file $f$, it will start to download encoded segments of file $f$ from the encountered users, and also check its own cache. We assume that the duration of each contact of users $i$ and $j$ is $t^c_{i,j}$ seconds, and the transmission rate from user $j$ to user $i$ is $r_{i,j}$ bps. Accordingly, we consider that $B_{i,j}=\left\lfloor \frac{t^c_{i,j}r_{i,j}}{s} \right\rfloor$ segments can be transmitted within one contact from user $j$ to user $i$, which is more practical than most existing works, e.g., \cite{mobilityD2Dcaching}. We also assume that there is a delay constraint, denoted as $T^d$. If user $i$ cannot collect at least $K_f$ different encoded segments of file $f$ within $T^d$, it will request the remaining segments from the BS. For example, as shown in Fig. \ref{model}, user $1$ requests a file, which is not in its own cache. Thus, it starts to wait for encountering the users storing the requested file. Within time $t<T^d$, user $1$ first gets one segment from user $2$, and then another segment from user $3$. As  user $1$ collects enough segments via D2D links to recover the requested file, it does not need to download the file from the BS.

\subsection{Problem Formulation}
In this paper, we investigate the caching placement strategy to maximize the \emph{data offloading ratio}, i.e., the percentage of the requested data that can be delivered via D2D links, rather than via the BS. If the D2D links can offload more data from the BS, it will lead to higher spatial reuse efficiency and also significantly reduce the backhaul burden. Specifically, for user $i$, the data offloading ratio is defined as 
\begin{equation} 
E_i=\mathbb{E}_{f \in \mathcal{F}} \left[ \mathbb{E}_{u_{i,f}} \left[ \left. \frac{\min \left(u_{i,f}, K_f \right)}{K_f} \right| \text{User $i$ requests file $f$}\right] \right]=
\sum \limits_{f \in \mathcal{F}} p_f \left\{\frac{ \mathbb{E}_{u_{i,f}}\left[ \min \left(u_{i,f}, K_f \right) \right]}{K_f}\right\},
\end{equation}
where $u_{i,f}=\sum_{j \in \mathcal{D}} u^j_{i,f}$ denotes the number of encoded segments of file $f$ that can be collected by user $i$ within time $T^d$, of which $u^j_{i,f}$ segments can be collected from user $j$. Let $M_{i,j}$, with $i \in \mathcal{D}$ and $j \in \mathcal{D}$, denote the number of contact times for user $i$ and $j$ within time $T^d$. Based on the memoryless property of the Poisson process, $M_{i,j}$ follows a Poisson distribution with mean $\lambda_{i,j}T^d$. Since $B_{i,j}$ segments can be transmitted from user $j$ in one contact time, user $i$ can maximally download $B_{i,j}M_{i,j}$ segments from user $j$ within $T^d$. Considering the cached content at user $j$, user $i$ will get $u^j_{i,f}=\min (B_{i,j}M_{i,j},x_{j,f})$ encoded segments of file $f$ from user $j$. The pairwise contact rate $\lambda_{i,i}$ can be regarded as infinity, which implies that user $i$ can access all the contents in its own cache. Thus, during time $T^d$, user $i$ can collect $u_{i,f}=\sum_{j \in \mathcal{D}} \min (B_{i,j} M_{i,j},x_{j,f})$ encoded segments of file $f$ from all users. Then, the data offloading ratio for user $i$ is
\begin{equation} \label{expectation}
E_i=\sum \limits_{f \in \mathcal{F}} \frac{p_f}{K_f} \left\{ \mathbb{E}\left[ \min \left(\sum_{j \in \mathcal{D}} \min (B_{i,j} M_{i,j},x_{j,f}), K_f \right) \right]\right\}.
\end{equation}
Accordingly, the overall average data offloading ratio of all the users is $\frac{1}{N_u} \sum\limits_{i \in \mathcal{D}} E_i$, and the mobility-aware caching placement problem is formulated as
\begin{align}
\mathop {\max }\limits_{X } \quad & \frac{1}{N_u} \sum\limits_{i \in \mathcal{D}} E_i ,\label{problem}
\end{align}
\begin{align}
\text{s.t.} \quad &\sum \limits_{f \in \mathcal{F}} x_{j,f} \leqslant C , \forall j \in \mathcal{D}, \tag{\ref{problem}a} \\
& x_{j,f} \in \mathbb{N}, \forall j \in \mathcal{D} \text{ and } f \in \mathcal{F}, \tag{\ref{problem}b} \\
& x_{j,f} \leqslant K_f, \forall j \in \mathcal{D} \text{ and } f \in \mathcal{F}, \tag{\ref{problem}c}
\end{align}
where constraint (\ref{problem}a) implies that each user device cannot store more than $C$ segments, (\ref{problem}b) guarantees that each encoded segment is either fully stored or not stored at all in one user device, and (\ref{problem}c) implies that, since file $f$ can be recovered by $K_f$ encoded segments, it is redundant to store more than $K_f$ segments of file $f$ at each user. The caching placement problem (\ref{problem}) is a mixed integer nonlinear programming (MINLP) problem. Next, we will show that it is NP-hard in Lemma \ref{NP}. 

\begin{lem} \label{NP}
	Problem (\ref{problem}) is an NP-hard problem.
\end{lem}
\begin{proof}
	See Appendix A. 	
\end{proof}

In the following sections, we will propose effective algorithms to solve it.

\subsection{A Divide and Conquer Algorithm To Evaluate the Objective}
To evaluate the objective function of the optimization problem (\ref{problem}), the distribution of the summation of $N_u$ independent random variables need to be obtained. A direct calculation will incur high computational complexity, which increases exponentially with the number of users $N_u$. To efficiently solve (\ref{problem}), we first develop a time-efficient approach to evaluate the objective. With the definition of the expectation, we can rewrite the objective function as
\begin{align} \label{ob}
\frac{1}{N_u} \sum\limits_{i \in \mathcal{D}} E_i =& \frac{1}{N_u}\sum \limits_{i \in \mathcal{D}}\sum \limits_{f \in \mathcal{F}} \frac{p_f}{K_f} \Bigg\{ \sum \limits_{q=0}^{K_f-1} q \Pr \left[  \sum_{j \in \mathcal{D}} \min (B_{i,j} M_{i,j},x_{j,f}) =q\right] \notag 
\end{align}
\begin{align}
&+K_f - K_f\Pr \left[ \sum_{j \in \mathcal{D}} \min (B_{i,j} M_{i,j},x_{j,f}) \leqslant K_f-1 \right] \Bigg\}.
\end{align}
Then, the main difficulty is to get the distribution of the summation of the $N_u$ independent random variables, denoted as $S^u_{i,f}=\sum_{j \in \mathcal{D}} \min (B_{i,j} M_{i,j},x_{j,f})$. In the following, we will propose a divide and conquer algorithm to reduce the computation complexity of getting the cumulative distribution function (cdf) of $S^u_{i,f}$. The algorithm decomposes the original problem into multiple sub-problems, and solve them one by one. Then, the solution of the original problem is constructed based on those for the sub-problems.

Firstly, the cdf of $S^u_{i,f}$ can be written as
\begin{equation} \label{cdf}
\text{Pr}\left[ \sum_{j = 1}^{N_u} V_{j}^{(i,f)} \leqslant K' \right],
\end{equation}
where $K'$ is an integer and $V_{j}^{(i,f)}=\min (B_{i,j}M_{i,j}, x_{j,f})$, with $1 \leqslant j \leqslant N_u$, is a non-negative random variable.
Since all the random variables are non-negative integer numbers, the event $K'<0$ is with zero probability. In the following, we only consider the case where $K'$ is a non-negative integer number.

Secondly, we define a series of sub-problems, denoted as $P(J,K^r)$, where $0 \leqslant J \leqslant N_u$, $0 \leqslant K^r \leqslant K'$, and both $J$ and $K^r$ are integer numbers. Each sub-problem $P(J,K^r)$ calculates the probability that the summation of random variables $V_{j}^{(i,f)}$, with $1 \leqslant j \leqslant J$, is less than or equal to $K^r$, which is given as
\begin{equation} \label{sub-problem}
P(J,K^r)=\text{Pr} \left[ \sum_{j = 1}^{J} V_{j}^{(i,f)} \leqslant K^r \right].
\end{equation}

Finally, we will show that the probability given in (\ref{cdf}) can be calculated by solving multiple sub-problems defined in (\ref{sub-problem}). When $J=0$, no random variable is considered in the sub-problem, which becomes $P(J,K^r)=\Pr[0 \leqslant K^r]=1$. When $J \geqslant 1$, the sub-problem $P(J,K^r)$ can be solved based on the solutions of $P(J-1,\cdot)$, as shown in lemma \ref{lem_re}.
\begin{lem} \label{lem_re}
The sub-problem $P(J,K^r)$, with $1 \leqslant J \leqslant N_u$, $0 \leqslant K^r \leqslant K'$, can be constructed by
\begin{equation} \label{sub_re}
P(J,K^r)=\sum \limits_{z=0}^{K^r}
\Pr \left[V_{J}^{(i,f)}=z \right] \times P(J-1,K^r-z),
\end{equation}
where
\begin{equation} \label{sub_probability}
\Pr\left[V_{J}^{(i,f)}=z\right]=
\begin{cases}
\frac{(\lambda_{i,J}T^d)^{\frac{z}{B_{i,J}}} \exp (-\lambda_{i,J} T^d)}{\left(\frac{z}{B_{i,J}} \right)!} & \text{ if } z < x_{J,f} \text{ and } z \equiv 0 \text{ } \left( \text{\emph{mod} } B_{i,J} \right), \\
1-\frac{\Gamma \left( \left\lceil \frac{z}{B_{i,J}}\right\rceil, \lambda_{i,J} T^d \right) }{\left( \left\lceil \frac{z}{B_{i,J}}\right\rceil-1 \right)!} & \text{ if } z= x_{J,f}, \\
0 & \text{ otherwise. }
\end{cases}
\end{equation}
\end{lem}
\begin{proof}
Eq. (\ref{sub_re}) follows the law of total probability, and the probability $\Pr \left[ V_J^{(i,f)} =z \right]$ equals
\begin{equation} \label{sub_probability2}
\Pr\left[V_{J}^{(i,f)}=z\right]=\Pr \left[ \min (BM_{i,J}, x_{J,f})=z \right]=
\begin{cases}
\Pr \left[ M_{i,J}= \frac{z}{B_{i,J}} \right] &  \text{ if } z < x_{J,f}, \\
\Pr \left[ M_{i,J} \geqslant \frac{z}{B_{i,J}}  \right] &  \text{ if } z = x_{J,f}, \\
0 & \text{ if } z > x_{J,f}.
\end{cases}
\end{equation}
Since $M_{i,J}$ is a Poisson random variable, we can easily get the result in (\ref{sub_probability}).
\end{proof}
Accordingly, when $J=N_u$, the cdf of $S^u_{i,f}$ can be obtained. The divide and conquer algorithm is summarized in Algorithm \ref{alg:dc}. Using the cdf of $S^u_{i,f}$, we can get the objective value based on (\ref{ob}). The overall complexity of evaluating the objective function is $\mathcal{O}(N_{file} N_u^2 K_{\max}^2)$, where $K_{\max}=\max \limits_{f \in \mathcal{F}} K_f$. Compared to calculating the objective directly using multiple summations, whose complexity is $\mathcal{O}(N_{file} N_u (K_{\max})^{N_u})$, the proposed divide and conquer algorithm significantly reduces the complexity.
\begin{algorithm}[!t]
\caption{The divide and conquer algorithm to calculate $\text{Pr}\left[ \sum_{j = 1}^{N_u} V_{j}^{(i,f)} \leqslant K' \right]$}
\label{alg:dc}
\begin{algorithmic}[1]
\FORALL{$0 \leqslant K^r \leqslant K'$}
\STATE {Set $P(0,K^r)=1$.}
\ENDFOR
\FOR{$J=1$ to $N_u$}
\FORALL{$0 \leqslant K^r \leqslant K'$}
\STATE {Set $P(J,K^r)=0$.}
\FORALL{$0 \leqslant z \leqslant K^r$}
\STATE Set $P(J,K^r) = P(J,K^r) + \Pr \left[V_{J}^{(i,f)}=z \right] \times P(J-1,K^r-z)$.
\ENDFOR
\ENDFOR
\ENDFOR
\end{algorithmic}
\end{algorithm}

\section{Optimal Mobility-Aware Caching Algorithm}
In this section, an optimal DP algorithm will be proposed to solve the mobility-aware caching placement problem. It will help to demonstrate the effectiveness of the proposed caching strategy and also serve as a performance benchmark to evaluate the sub-optimal algorithm proposed in Section IV.

\subsection{Optimal Caching Placement Algorithm}
\begin{figure}[!t]
  \centering
  \includegraphics[width=6.5in]{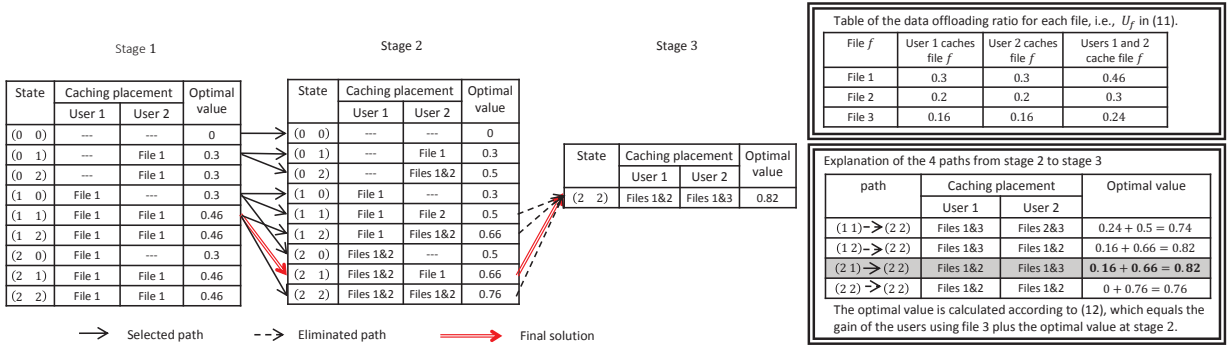}
  \caption{An illustration of the proposed DP algorithm with $N_u=2$, $N_f=3$ and $K_f=1$. There are $3$ stages, where the eliminated paths from stage 1 to stage 2 are omitted. The states at stage $m$ represent the remaining caching capacity to store the first $m$ files. For example, state $(1 \quad 2)$ at stage $2$ means that user 1 caches one segment of file 1 or 2, and user 2 caches two segments of files 1 and 2 in total. The final solution first caches file 1 at both users 1 and 2, then caches file 2 at user 1, finally caches file 3 at user 2.}
  \label{DP-stage}
\end{figure}
DP is an effective algorithm to deal with integer variables, which divides the original problem into multiple stages, and solves the problem stage by stage \cite{dpint}. Specifically, each stage is associated with multiple states, and the optimal solution at each stage can be obtained based on the optimal solutions at previous stages according to a recursive relationship. Thus, with the optimal solution at the beginning stage and the recursive relationship, the optimal solution of the original problem can be constructed. In the following, we will illustrate the key steps of the proposed DP algorithm for the caching placement problem, including dividing the original problem into stages and finding the corresponding recursive relationship.
\subsubsection{Stages and States}
At each stage, we will design caching placement for a specific file. Each state includes the information of the remaining cache capacities for all the mobile users. We denote the state as a vector $\mathbf{C^r}=[c^r_j]_{1 \leqslant j \leqslant N_u}$, where $ c^r_j$ is the remaining cache capacity of user $j$.

For state $\mathbf{C^r}$ at the $m$-th stage, $1 \leqslant m \leqslant N_{file}$, the optimal caching placement for the first $m$ files while user $i$ stores no more than $c^r_i$ segments needs to be found. It can be obtained by solving the following problem:
\begin{align}
OPT_{m,\mathbf{C^r}}=\mathop {\max }\limits_{x_{j,f}, j \in \mathcal{D}, 1 \leqslant f \leqslant m } \quad & \sum \limits_{1 \leqslant f \leqslant m} U_f(x_{1,f},x_{2,f}, \dots , x_{N_u,f} ) ,\label{problem_dp}\\
\text{s.t.} \quad &\sum \limits_{0 \leqslant f \leqslant m} x_{j,f} \leqslant c^r_j , \forall j \in \mathcal{D}, \tag{\ref{problem_dp}a}
\end{align}
\begin{align}
& x_{j,f} \in \mathbb{N}, \forall j \in \mathcal{D} \text{ and } 1 \leqslant f \leqslant m, \tag{\ref{problem_dp}b} \\
& x_{j,f} \leqslant K_f, \forall j \in \mathcal{D} \text{ and } 1 \leqslant f \leqslant m, \tag{\ref{problem_dp}c}
\end{align}
where
\begin{equation} 
U_f(x_{1,f},x_{2,f}, \dots , x_{N_u,f} )= \frac{1}{N_u} \sum\limits_{i \in \mathcal{D}} \frac{p_f}{K_f} \left\{ \mathbb{E}\left[ \min \left(\sum \limits_{j \in \mathcal{D}} \min (B_{i,j} M_{i,j},x_{j,f}) , K_f \right) \right]\right\}
\end{equation}
is a function from $N_u$ numbers to a real number representing the data offloading ratio for file $f$, and constraints (\ref{problem_dp}a)-(\ref{problem_dp}c) correspond to constraints (\ref{problem}a)-(\ref{problem}c) respectively. For state $[C]_{N_u \times 1}$ at the $N_{files}$-th stage, problem (\ref{problem_dp}) is the same as the original problem (\ref{problem}). Denote the optimal caching placement at the $m$-th stage with state $\mathbf{C^r}$ as $(x_{m,\mathbf{C^r}})_{j,f}$, which means that user $j$ stores $(x_{m,\mathbf{C^r}})_{j,f}$ segments of file $f$. This problem is still complicated, and next we will develop a recursive relationship to solve it.

\subsubsection{Recursive Relationship}
At the first stage, since there is only one file to cache, the optimal solution is to fill the remaining caching storage of each user with the first file, i.e., $OPT_{1,\mathbf{C^r}}=U_1(\min(K_1,c^r_1), \min(K_1,c^r_2), \dots , \min (K_1,c^r_{N_u}) )$. Then, the optimal solutions at the $m$-th stage, with $2 \leqslant m \leqslant N_{file}$, can be constructed based on the optimal solutions at the $(m-1)$-th stage as shown in Lemma \ref{lem_dp}, and an illustrate is shown in Fig. \ref{DP-stage}.
\begin{lem} \label{lem_dp}
The optimal solution at the $m$-th stage, $2 \leqslant m \leqslant N_{file}$, with state $\mathbf{C^r}$ can be obtained according to the following recursive relationship:
\begin{align}
OPT_{m,\mathbf{C^r}}=\mathop {\max }\limits_{x_{j,m}, j \in \mathcal{D} } \quad & U_m(x_{1,m},x_{2,m}, \dots , x_{N_u,m} )+OPT_{m-1,[c_j^r-x_{j,m}]_{1 \leqslant j \leqslant N_u}} ,\label{problem_re} \\
\text{s.t.} \quad & x_{j,m} \leqslant c^r_j , \forall j \in \mathcal{D}, \tag{\ref{problem_re}a} \\
& x_{j,m} \in \mathbb{N}, \forall j \in \mathcal{D}, \tag{\ref{problem_re}b} \\
& x_{j,m} \leqslant K_f, \forall j \in \mathcal{D}. \tag{\ref{problem_re}c}
\end{align}
\end{lem}
\begin{proof}
We first consider a certain caching placement for file $m$, i.e., user $j$ stores $x_{j,m}$ segments of file $m$. Then, the remaining cache capacities for the first $(m-1)$ files becomes $[c_j^r-x_{j,m}]_{1 \leqslant j \leqslant N_u}$. Meanwhile, the objective value
can be calculated based on $OPT_{m-1,[c_j^r-x_{j,m}]_{1 \leqslant j \leqslant N_u}}$, given as
\begin{align}
&\mathop {\max }\limits_{x_{j,f}, j \in \mathcal{D}, 1 \leqslant f \leqslant m-1 } \sum \limits_{1 \leqslant f \leqslant m} U_f(x_{1,f},x_{2,f}, \dots , x_{N_u,f} ) \notag \\
=&U_m(x_{1,m},x_{2,m}, \dots , x_{N_u,m} )+\mathop {\max }\limits_{x_{j,f}, j \in \mathcal{D}, 1 \leqslant f \leqslant m-1 } \sum \limits_{1 \leqslant f \leqslant m-1} U_f(x_{1,f},x_{2,f}, \dots , x_{N_u,f} ) \notag \\
=&U_m(x_{1,m},x_{2,m}, \dots , x_{N_u,m} )+OPT_{m-1,[c_i^r-x_{i,m}]_{N_u}}.
\end{align}
Finally, after searching all possible caching placements for file $m$, we can get the recursive relationship in (\ref{problem_re}), where constraints (\ref{problem_re}a)-(\ref{problem_re}c) correspond to constraints (\ref{problem_dp}a)-(\ref{problem_dp}c) respectively.
\end{proof}

\subsubsection{Optimal Caching Placement}
Denote $(x^\star_{1,m}, x^\star_{2,m}, \dots, x^\star_{N_u,m})$ as the optimal solution of problem (\ref{problem_re}), the optimal caching placement at the $m$-th stage with state $\mathbf{C^r}$ is decided as follows. When $m=1$, all the users cache the first file, and thus, the caching placement is
\begin{equation} \label{placement}
(x_{1,\mathbf{C^r}})_{j,f}=
\begin{cases}
\min (K_1, c^r_j) & \text{ if } f=1, \\
\text{undecided} & \text{ if } f>1.
\end{cases}
\end{equation}
When $m>1$, the caching placement for file $m$ is $(x^\star_{1,m},x^\star_{2,m}, \dots, x^\star_{N_u,m})$, and for the first $(m-1)$ files, it is the same as the optimal solution at the $(m-1)$-th stage with state $[c_j^r-x^\star_{j,m}]_{1 \leqslant j \leqslant N_u}$, given as
\begin{equation} \label{placement}
(x_{m,\mathbf{C^r}})_{j,f}=
\begin{cases}
x^\star_{j,m} & \text{ if } f=m, \\
(x_{m-1,[c_j^r-x^\star_{j,m}]_{1 \leqslant j \leqslant N_u}})_{j,f} & \text{ if } f<m, \\
\text{undicided} & \text{ if } f>m.
\end{cases}
\end{equation}
The optimal solution at the first stage can be obtained by filling the caching storage of each user with the first file, and then, each stage can be solved based on the solutions at previous stages. Thus, we can get the optimal solution for the original problem at the $N_{file}$-th stage. The optimal DP algorithm is summarized in Algorithm \ref{alg:dp}.

\subsection{Complexity Analysis}

According to Algorithm \ref{alg:dp}, the complexity of obtaining the optimal solution at the $m$-th stage with state $\mathbf{C^r}$, when $1<m \leqslant N_{file}$, is given as
\begin{equation} \label{tc_dp1}
\begin{aligned}
\mathcal{T}_{state}=\mathcal{O} \left((K_{\max}+1)^{\lVert \mathbf{C^r} \lVert_0 } T_u \right),
\end{aligned}
\end{equation}
where $T_u$ is the time complexity of calculating the function $U_m(x_{1,m},x_{2,m}, \dots , x_{N_u,m} )$, which equals $\mathcal{O}(N_u^2 K_{\max}^2)$ using the divide and conquer algorithm, and $\lVert \cdot \lVert_0$ denotes the number of non-zero elements in the vector. The time complexity of the $m$-th stage can then be derived as
\begin{equation} \label{tc_dp2}
\mathcal{T}_{stage}(m)=
\begin{cases}
\mathcal{O} ((C+1)^{N_u} N_u^2 K_{\max}^2) & \text{ if } m=1, \\
\mathcal{O} \left((C K_{\max}+C+1)^{N_u} N_u^2 K_{\max}^2\right) & \text{ if } 1<m<N_{file}, \\
\mathcal{O} \left((K_{\max}+1)^{N_u} N_u^2 K_{\max}^2\right) & \text{ if } m=N_{file}.
\end{cases}
\end{equation}
Then, we can get the time complexity of the DP algorithm as
\begin{equation} \label{tc_dpall}
\mathcal{T}_{dp}=
\begin{cases}
\mathcal{O} \left( ((C+1)^{N_u}+(K_{\max}+1)^{N_u}) N_u^2 K_{\max}^2 \right) & \text{ if } N_{file} \leqslant 2 \\
\mathcal{O} \left(N_{file}(C K_{\max}+C+1)^{N_u} N_u^2 K_{\max}^2 \right) & \text{ if } N_{file}>2.
\end{cases}
\end{equation}
For comparison, the time complexity of the exhaustive search method is given by
\begin{equation} \label{tc_ex}
\mathcal{T}_{search}=\mathcal{O} \left(N_{file}(N_{file}K_{\max})^{CN_u} N_u^2 K_{\max}^2 \right).
\end{equation}

\begin{algorithm}[!t]
	\caption{The optimal DP algorithm}
	\label{alg:dp}
	\begin{algorithmic}[1]
		\STATE $OPT_{N_{file},[C]_{N_u \times 1}}, \left( x_{N_{file},[C]_{N_u \times 1}} \right)_{j,f}$: the optimal utility value and  caching placement for the original problem
		\FORALL{$\mathbf{C^r} \leqslant [C]_{N_u \times 1}$}
		\STATE {Set $(x_{1,\mathbf{C^r}})_{j,1}=\min (K_1, c^r_j)$ for $j \in \mathcal{D}$.}
		\STATE {Set $OPT_{1,\mathbf{C^r}}=U_1((x_{1,\mathbf{C^r}})_{1,1},(x_{1,\mathbf{C^r}})_{2,1}, \dots , (x_{1,\mathbf{C^r}})_{N_u,1} )$.}
		\ENDFOR
		\FOR{$m=2$ to $N_{file}-1$}
		\FORALL{$\mathbf{C^r} \leqslant [C]_{N_u \times 1}$}
		\STATE {Set $OPT_{m,\mathbf{C^r}}=0$.}
		\FORALL{$x_{j,m}$, $j \in \mathcal{D}$, that satisfies constraints in (\ref{problem_re})}
		\IF{$OPT_{m,\mathbf{C^r}} < U_m(x_{1,m},x_{2,m}, \dots , x_{N_u,m} )+OPT_{m-1,[c_j^r-x_{j,m}]_{1 \leqslant j \leqslant N_u}}$}
		\STATE Set $OPT_{m,\mathbf{C^r}} = U_m(x_{1,m},x_{2,m}, \dots , x_{N_u,m} )+OPT_{m-1,[c_j^r-x_{j,m}]_{1 \leqslant j \leqslant N_u}}$ \\ and $x_{j,m}^\star = x_{j,m}$.
		\ENDIF
		\ENDFOR
		\STATE Update $\left(x_{m,\mathbf{C^r}} \right)_{j,f}$, $j \in \mathcal{D}$, $f \in \mathcal{F}$, according to (\ref{placement}).
		\ENDFOR
		\ENDFOR
		\STATE {Set $m=N_{file}$,$\mathbf{C^r}=[C]_{N_u \times 1}$, and Run Steps 8-14.}
	\end{algorithmic}
\end{algorithm}

The only difference between the complexities of the DP algorithm and the exhaustive search is the bases of the exponential terms, which significantly affect the overall complexities. When $N_{file}=1$, the caching placement can be easily obtained by filling all the caching storage with the only file. When $N_{file} \geqslant 2$, it can be observed that $(N_{file} K_{\max})^C$ increases faster than $(C+1)$, $(K_{\max}+1)$ and $(C K_{\max} +C +1)$ with both $K_{\max}$ and $C$ increasing. Meanwhile, considering the case while both $K_{\max}$ and $C$ equal the minimum value, i.e., $K_{\max}=C=1$, the complexity of the DP algorithm is similar to that of the exhaustive search when $N_{file} \leqslant 3$, and much faster than the later when $N_{file}>3$. Thus, the proposed DP algorithm is superior than the exhaustive search with respect to the time efficiency.

\section{Submodular Optimization Based Algorithm for Mobility-Aware Caching}
The complexity of the DP algorithm increases exponentially with the number of mobile users. Thus, in this section, we aim to design a more practical sub-optimal algorithm. We will first reformulate the original problem to the one that maximizes a monotone submodular function over a matroid constraint. Then, adopting a greedy algorithm, we will provide a near-optimal caching placement algorithm.

\subsection{Submodular Functions and Matroid Constraints}
The submodular maximization problem subject to matroid constraints is an important type of discrete optimization problems. It has been widely investigated, and effective algorithms with provable approximation ratios have been proposed \cite{submodularbook,submodular2,submodular}. A great number of typical and interesting combinational optimization problems are special cases of the submodular maximization problem, e.g., the max-k-cover problem and the max-cut problem \cite{1e}. Recently, submodular maximization problems over matroid constraints have found a wide range of applications in wireless networks, e.g., for network coding design \cite{networkcoding} and femto-caching design \cite{femtocaching2}.
In the following, we will first provide some background information, including necessary definitions and properties of the monotone submodular function and the matroid constraint \cite{submodular2}.
\begin{defn}
Let $S$ be a finite set and $g$ be a real-valued function with subsets of $S$ as its domain. The function $g$ is a monotone \emph{submodular function} if it satisfies the following properties:
\begin{enumerate}
  \item $g(A) \leqslant g(D)$, for all $A \subseteq D \subseteq S$,
  \item $g(A)+g(D) \geqslant g(A \cup D) + g(A \cap D)$.
\end{enumerate}
\end{defn}
\begin{prop}
The following statement also defines a monotone submodular function:
\begin{align}
g(A \cup \{j\}) - g(A) \geqslant g(D \cup \{j\}) - g(D) \geqslant 0,
\text{ for all } A \subset D \subset S \text{ and } j \in S-D.
\end{align}
\end{prop}
Proposition 1 provides an intuitive interpretation of the monotone submodular function, i.e., the gain of adding a new element decreases when the set becomes large. It is monotone as adding a new element always has a non-negative gain.
\begin{defn}
Let $S$ be a finite ground set and $\mathcal{I}$ be a nonempty collection of subsets of $S$, and then, the tuple $\mathcal{M}=\{S,\mathcal{I}\}$ is called a \emph{matroid} if
\begin{enumerate}
  \item $\varnothing \in \mathcal{I}$,
  \item If $D \in \mathcal{I}$ and $A \subseteq D$, then $A \in \mathcal{I}$,
  \item If $A,D \in \mathcal{I}$ and $|A|<|D|$, then there exists $j \in D-A$ such that $A \cup \{j\} \in \mathcal{I}$,
\end{enumerate}
where $| \cdot |$ denotes the cardinality.
\end{defn}

\subsection{Problem Reformulation}
In the following, we will show that the original problem (\ref{problem}) can be reformulated as a monotone submodular maximization problem over a matroid constraint. Firstly, the ground set $S$ is defined as $S=\{y_{j,f,k}| j \in \mathcal{D}, f \in \mathcal{F} \text{ and } 1 \leqslant k \leqslant K_f \}$, and each caching placement $Y$ is a subset of $S$. If the element $y_{j,f,k}$ is in $Y$, then it means that the $k$-th segment of file $f$ is cached at user $j$. Specifically, the relationship between $x_{j,f}$, with $ j \in \mathcal{D}$ and $f \in \mathcal{F}$, and the caching placement $Y \subseteq S$ is
\begin{equation}
x_{j,f}=|Y \cap S_{j,f}|,
\end{equation}
where $S_{j,f}$ is a set of all $K_f$ segments of file $f$ that may be cached at user $j$ defined as
\begin{equation}
S_{j,f}=\{y_{j,f,k}| 1 \leqslant k \leqslant K_f \}.
\end{equation}
Next, we rewrite the objective function of problem (\ref{problem}) as a function of subsets of $S$ as
\begin{equation} \label{objective}
g(Y)=\frac{1}{N_u}\sum \limits_{i \in \mathcal{D}}\sum \limits_{f \in \mathcal{F}} \frac{p_f}{K_f} \left\{ \mathbb{E}\left[ \left( \sum_{j \in \mathcal{D}} \min (B_{i,j}M_{i,j},|Y \cap S_{j,f}|),K_f \right) \right]\right\}.
\end{equation}
To prove that the function $g(Y)$ defined in (\ref{objective}) is a monotone submodular function, we will show that $g(Y)$ satisfies Proposition 1. The main difficulty is to derive the gain of adding one element into the caching placement $Y$, which is given in Lemma \ref{lem_gain}.
\begin{lem} \label{lem_gain}
Let $Y \subset S \text{ and } y_{j,f,k} \in S-Y$, and then the gain of adding $y_{j,f,k}$ into the caching placement $Y$ is
\begin{align}
g(Y \cup \{y_{j,f,k}\})-g(Y) =&
\frac{p_f }{N_u K_f}  \sum \limits_{i \in \mathcal{D}} \Pr \bigg[ M_{i,j} \geqslant \left\lceil \frac{|Y \cap S_{j,f}|+1}{B_{i,j}}\right\rceil \bigg] \notag \\
& \times \Pr \bigg[\sum_{j' \in \mathcal{D}, j' \ne j} \min (B_{i,j'}M_{i,j'}, |Y \cap S_{j',f}|) \le K_f- |Y \cap S_{j,f}| -1 \bigg]  \label{gain}
\end{align}
\end{lem}
\begin{proof}
See Appendix B.
\end{proof}
Lemma \ref{lem_sub} verifies that  $g(Y)$ satisfies Proposition 1, and thus it is a monotone submodular function.
\begin{lem} \label{lem_sub}
Let $ A \subset D \subset S \text{ and } y_{j,f,k} \in S-D$, then,
\begin{align}
g(A \cup \{y_{j,f,k}\}) - g(A) \geqslant g(D \cup \{y_{j,f,k}\}) - g(D) \geqslant 0
\end{align}
\end{lem}
\begin{proof}
See Appendix C.
\end{proof}
Finally, the constraint in the original problem is proved to be a matroid constraint in Lemma \ref{lem_mat}.
\begin{lem} \label{lem_mat}
Let $S_i$, where $i \in \mathcal{D}$, denote all the segments that may be stored at user $i$, which is $S_i=\{y_{i,f,k}| f \in \mathcal{F} \text{ and } 1 \leqslant k \leqslant K_f \}$. Then, constraints (\ref{problem}a) and (\ref{problem}b) can be rewritten as $Y \in \mathcal{I}$, where
\begin{equation} \label{I}
\mathcal{I}=\left\{Y \subseteq S \big| |Y \cap S_i| \leqslant C, \forall i \in \mathcal{D} \right\},
\end{equation}
which is a matroid constraint.
\end{lem}
\begin{proof}
The tuple $\mathcal{M}=(S,\mathcal{I})$ belongs to a typical matroid, called the partition matroid \cite{submodular2}.
\end{proof}
To sum up, problem (\ref{problem}) can be reformulated as
\begin{align}
\max \limits_{Y \in \mathcal{I}} \quad & \frac{1}{N_u}\sum \limits_{i \in \mathcal{D}}\sum \limits_{f \in \mathcal{F}} \frac{p_f}{K_f} \left\{ \mathbb{E}\left[\min  \left( \sum_{j \in \mathcal{D}} \min (B_{i,j}M_{i,j},|Y \cap S_{j,f}|),K_f \right) \right]\right\} \label{problem_m}
\end{align}

which is a monotone submodular maximization problem over a matroid constraint.

\subsection{Greedy Algorithm for Mobility-Aware Caching}
For problem (\ref{problem_m}), the greedy algorithm provides an effective solution, and has been proved to archive a $\frac{1}{2}$-approximation \cite{submodular2}, i.e., the worst case is at least $50\%$ of the optimal solution. Moreover, it has been observed to provide close-to-optimal performance in lots of cases. A randomized algorithm was proposed in \cite{1e}, which achieves a higher approximation ratio. However, with the size of the ground set as $N_u \times N_{file} \times K_{\max}$, the randomized algorithm is computationally inapplicable. In the simulation, we will show that the greedy algorithm performs very close to the optimal caching strategy.

The greedy algorithm firstly sets the caching placement $Y$ as an empty set, and then, it adds an element according to the priority values while satisfying the matroid constraint. The process continues until no more element can be added. The priority value of an element $y_{j,f,k}$, denoted as $g^d_{j,f,k}$, is defined as the gain of adding $y_{j,f,k}$ into the caching placement $Y$, given in (\ref{gain}), when $Y \in \mathcal{I}$, $y_{j,f,k} \in S-Y$ and $Y \cup \{y_{j,f,k}\} \in \mathcal{I}$.

The difficulty in the greedy algorithm is to update the priority value efficiently. To get the priority value given in (\ref{gain}), two kinds of probabilities need to be calculated, i.e.,
\begin{align}
p^1_{i,j,f}=&\Pr \bigg[ M_{i,j} \geqslant \left\lceil \frac{|Y \cap S_{j,f}|+1}{B_{i,j}}\right\rceil \bigg],\\
p^2_{i,j,f}=&\Pr \bigg[\sum_{j' \in \mathcal{D}, j' \ne j} \min (B_{i,j'}M_{i,j'}, |Y \cap S_{j',f}|)   \le K_f- |Y \cap S_{j,f}|-1 \bigg].
\end{align}
Since $M_{i,j}$ follows a Poisson distribution with mean $\lambda_{i,j}T^d$, $p^1_{i,j,f}$ can be calculated as
\begin{equation}
p^1_{i,j,f}=
1-\frac{\Gamma \left( \left\lceil \frac{|Y \cap S_{j,f}|+1}{B_{i,j}}\right\rceil, \lambda_{i,j} T^d \right) }{\left( \left\lceil \frac{|Y \cap S_{j,f}|+1}{B_{i,j}}\right\rceil-1 \right)!}.
\end{equation}
Using the divide and conquer algorithm, similar as the one proposed in Section II-D, we can get the value of $p^2_{i,j,f}$ efficiently.

\begin{algorithm}[!h]
\caption{The Greedy Algorithm}
\label{alg:greedy}
\begin{algorithmic}[1]
\STATE {Set $Y=\varnothing \Leftrightarrow \text{ set }x_{j,f}=0, \forall j \in \mathcal{D} \text{ and } f \in \mathcal{F}$.}
\STATE {$S^r=S$.}
\STATE {Initialize the priority values $\{ g^d_{j,f,k}| j \in \mathcal{D}, f \in \mathcal{F}, 1 \leqslant k \leqslant K_f \}$.}
\WHILE{$|Y| < N_u \times C $}
\STATE {$y_{j^\star,f^\star,k^ \star}=\arg \max \limits_{y_{j,f,k} \in S^r}  g^d_{j,f,k}$.}
\STATE {Set $Y=Y \cup \{y_{j^\star,f^\star, k^\star}\} \Leftrightarrow \text{ add } x_{j^\star,f^\star}$ by $1$.}
\STATE {$S^r=S^r-\{y_{j^\star,f^\star,k^\star}\}$.}
\IF {$|Y \cap S_{j^\star}|=C$}
\STATE {$S^r=S^r-S_{j^\star}$}
\ENDIF
\STATE {Update the priority values $\{ g^d_{j,f^\star,k} (Y)| j \in \mathcal{D}, 1 \leqslant k \leqslant K_{f^{\star}}, y_{j,f^\star,k} \in S^r \}$ based on Algorithm \ref{alg:dc}.}
\ENDWHILE
\end{algorithmic}
\end{algorithm}
In summary, the greedy algorithm is described in Algorithm \ref{alg:greedy}, where $S^r$ denotes the remaining set, including the elements that may be added into the caching placement $Y$. When the element $\{y_{j^\star,f^\star,k^\star}\}$ is added to $Y$, it should be deleted from $S^r$. When user $j^\star$ has already stored $C$ segments, no more segment can be stored in user $j^\star$, and thus, all the elements in $S_{j^\star}$ should be deleted from $S^r$. Moreover, after adding a new element $\{y_{j^\star,f^\star, k^\star} \}$ to $Y$, the priority values should be updated. According to (\ref{gain}), adding the element $\{y_{j^\star,f^\star, k^\star}\}$ only affects the priority values of elements $\{ y_{j,f^\star,k}| j \in \mathcal{D}, 1 \leqslant k \leqslant K_{f^{\star}}, y_{j,f^\star,k} \in S^r \}$, and thus, only the corresponding priority values need to be updated.

\subsection{Complexity Analysis}
According to Algorithm \ref{alg:greedy}, the computation complexity of the greedy algorithm is
\begin{equation}
\mathcal{O}\left(N_uC(T_p+T_m)\right),
\end{equation}
where $\mathcal{O}(T_p)$ and $\mathcal{O}(T_m)$ denote the computation complexity for updating the priority values and searching for the element $y_{j^\star,f^\star,k^\star}$ with the maximum priority value in one iteration. From Eq. (\ref{gain}), the priority value $g^d_{j,f,k} (Y)$ is irrelevant with the fact whether $y_{j,f^\star,k}$ is in $Y$ when $f \ne f^\star$, and the gain of adding different segments of the same file to the cache of the same user remains the same. Thus, in each iteration, at most $N_u$ priority values need to be calculated. According to Algorithm \ref{alg:dc}, the complexity for calculating one priority value is $\mathcal{O}(N_u^2  K^2_{\max})$. Then, the computation complexity for updating the priority values in one iteration is $\mathcal{O}(T_p)=\mathcal{O}(N_u^3 K^2_{\max})$. Moreover, to reduce the complexity of finding the element with the maximum priority value, we denote the maximum gain of caching one segment of file $f \in \mathcal{F}$ in one device as
\begin{equation}
max_{f}=\max \limits_{j \in \mathcal{D}, 1 \leqslant k \leqslant K_f, y_{j,f,k} \in S^r} g^d_{j,f,k} (Y).
\end{equation}
Since only priority values $g^d_{j,f^\star,k} (Y)$, where $j \in \mathcal{D}$ and $1 \leqslant k \leqslant K_{f^{\star}}$, may be updated in each iteration, we can store the maximum values $max_{f}$ for all $f \in \mathcal{F}$, and update $max_{f^\star}$ in each iteration to reduce the complexity for searching the element $y_{j^\star,f^\star,k^\star}$ with the maximum priority value from $\mathcal{O}(N_uN_{file})$ to $\mathcal{O}(N_u+N_{file})$. Thus, ignoring lower order terms, the overall computation complexity of the greedy algorithm is
\begin{equation}
\mathcal{O}\left(N_uC(N_u^3 K^2_{\max}+N_{file})\right).
\end{equation}

\section{simulation results}
In this section, simulation results will be provided to evaluate the performance of the proposed mobility-aware caching strategy. Moreover, we will validate the mobility model and the effectiveness of the greedy caching strategy based on real-life data-sets, which are available at CRAWDAD \cite{cambridge-haggle-20090529}. Four caching strategies are compared:
\begin{enumerate}
  \item{Optimal mobility-aware caching strategy:} the optimal solution of problem (\ref{problem}) using the DP algorithm in Section III.
  \item{Greedy mobility-aware caching strategy:} the sub-optimal solution of problem (\ref{problem}) using the greedy algorithm in Section IV.
  \item{Random caching strategy:} the random caching strategy \cite{randomcache}, where the probabilities of each user to cache segments of different files are proportional to the file request probabilities.
  \item{Popular caching strategy:} each user device stores $C$ segments of the most popular files \cite{pupolarcaching}.
\end{enumerate}
In the following, we assume that the file request probability follows a Zipf distribution with parameter $\gamma_r$, i.e., $p_f=\frac{f^{-\gamma_r}}{\sum \limits_{i \in \mathcal{F}} i^{-\gamma_r}}$ \cite{d2d-cache,femtocaching2,scaling}. The number of encoded segments to recover each file, i.e., $K_f$, is randomly selected in $[1,K_{\max}]$.
\subsection{Performance Evaluation}

\begin{figure}
\begin{minipage}[t]{0.31\linewidth}
\centering
\includegraphics[width=2.05in]{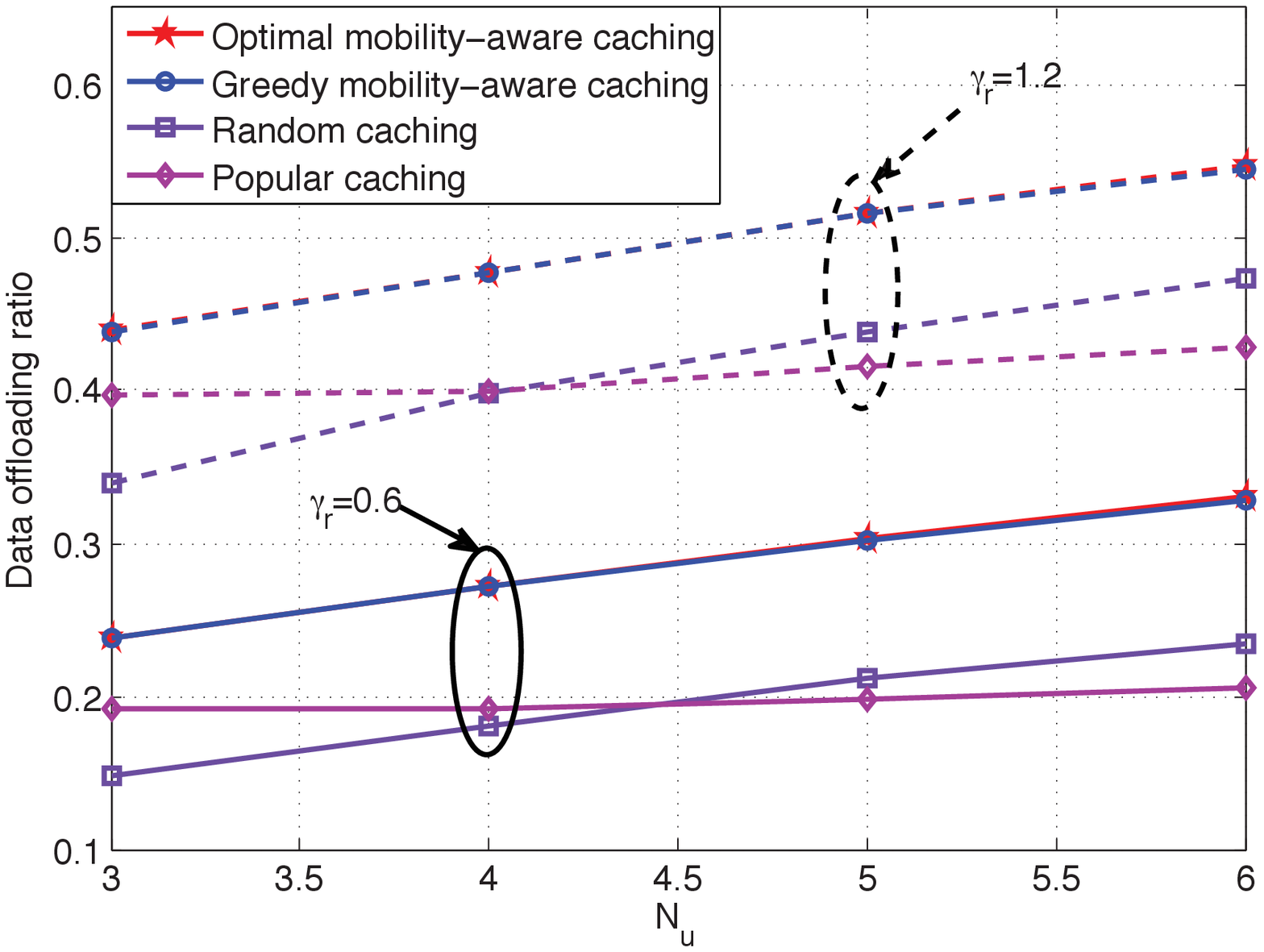}
  \caption{Comparison of different caching strategies with $N_{file}=20$, $T^d=120 \text{ s}$, $B_{i,j}=1$, $K_{\max}=3$ and $C=3$.}
\label{fig_scheme1}
\end{minipage}%
\hspace{0.1in}
\begin{minipage}[t]{0.31\linewidth}
\centering
  \includegraphics[width=2.05in]{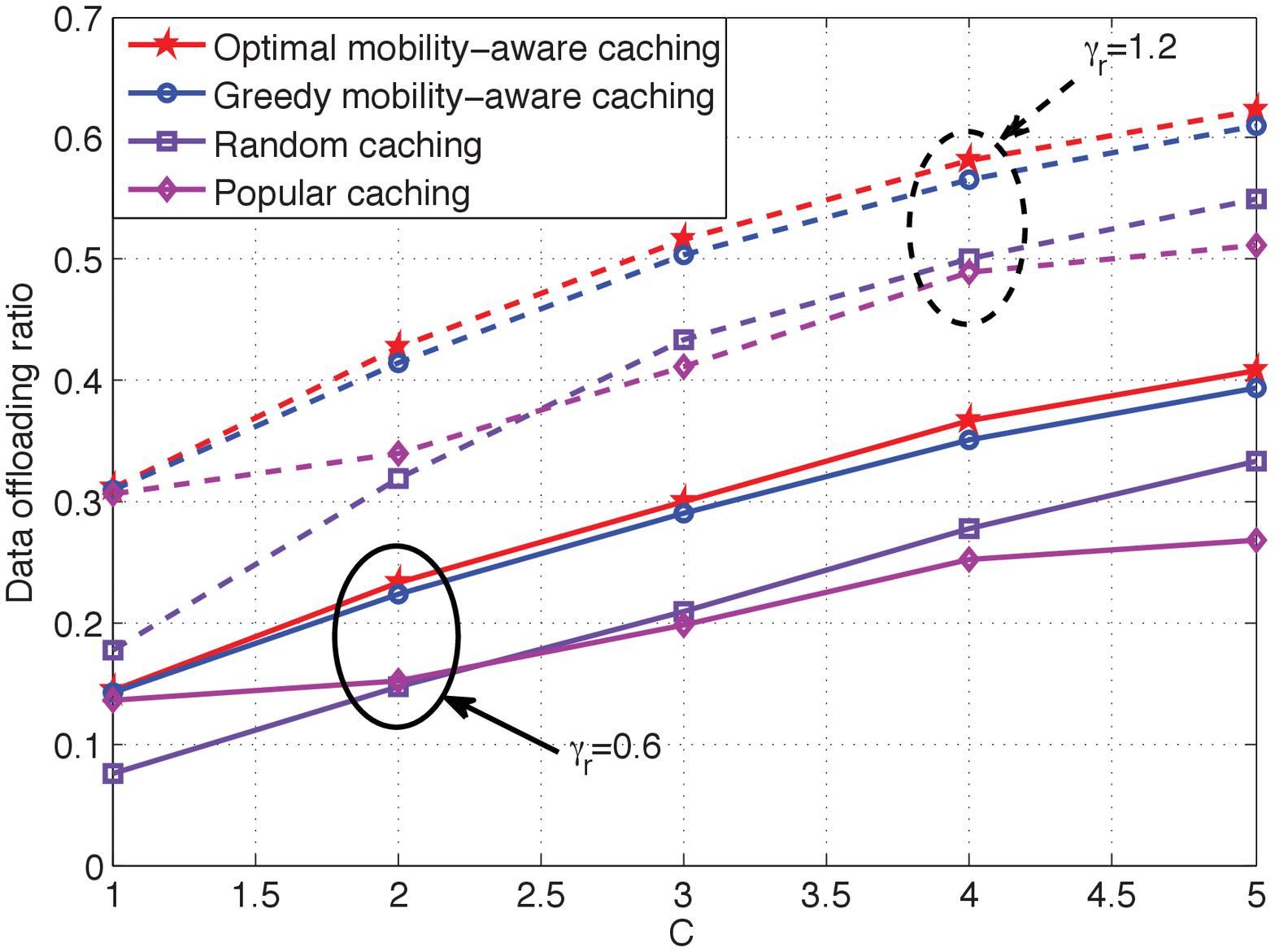}
  \caption{Comparison of different caching strategies with $N_u=5$, $N_{file}=20$, $T^d=120 \text{ s}$, $B_{i,j}=2$ and $K_{\max}=3$.}
\label{fig_scheme2}
\end{minipage}%
\hspace{0.1in}
\begin{minipage}[t]{0.31\linewidth}
\centering
  \includegraphics[width=2in]{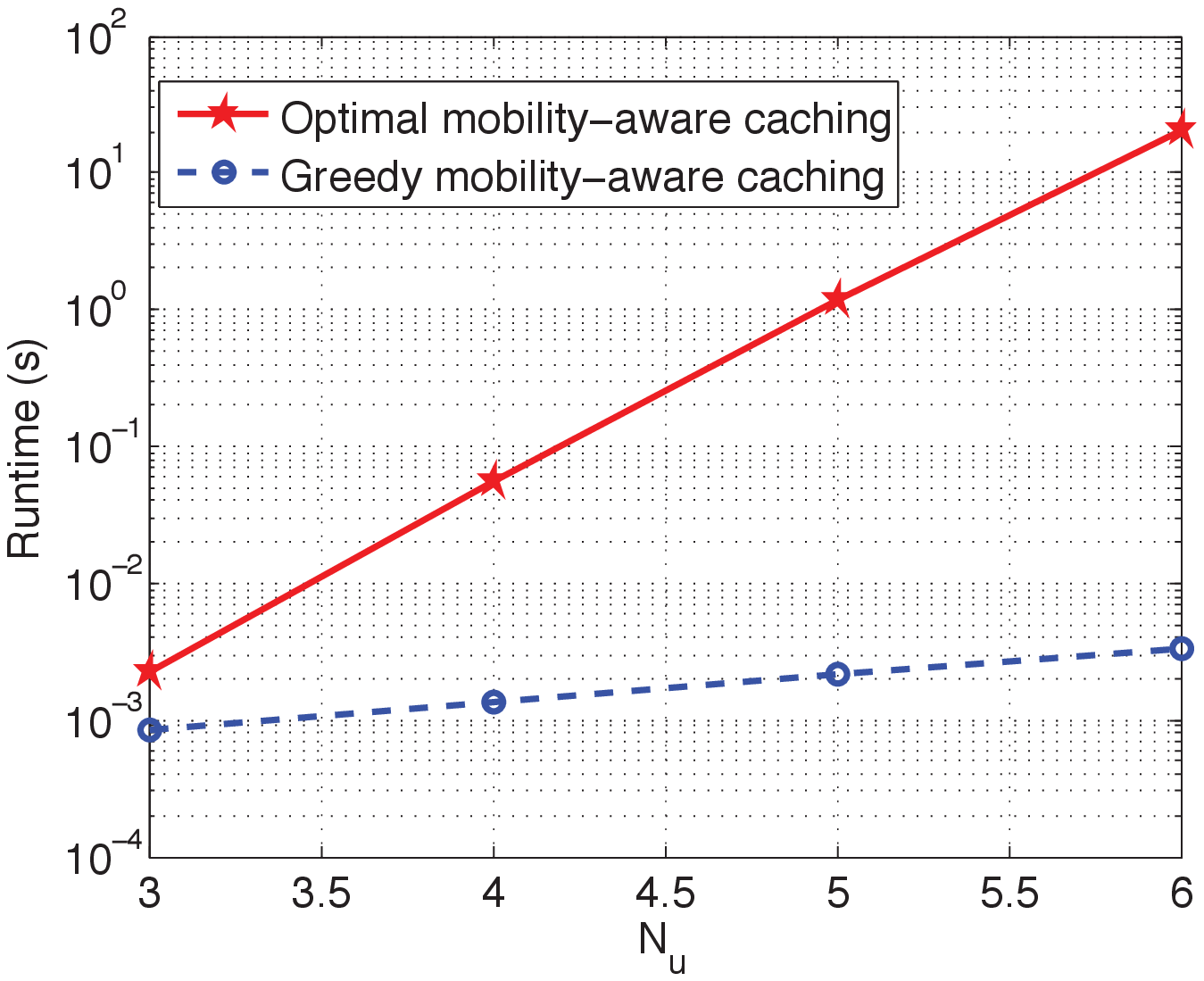}
		\caption{Runtime of the mobility aware caching strategy with $N_{file}=20$, $T^d=120 \text{ s}$, $B_{i,j}=1$, $\gamma_r=0.6$, $K_{\max}=3$ and $C=3$.}
\label{fig_runtime1}
\end{minipage}
\end{figure}

In this part, we generate the contact rate for each device pair, i.e., $\lambda_{i,j}$, for $i \in \mathcal{D}$ and $j \in \mathcal{D}\backslash \{i\}$, according to a Gamma distribution as $\Gamma(4.43,1/1088)$ \cite{aggregate_real}. We first compare the performance of the four caching strategies by varying the number of users in Fig. \ref{fig_scheme1}. Firstly, we can observe that the greedy caching strategy outperforms both the random caching and popular caching strategies, since the later two do not make good use of the user mobility information. Secondly, the performance of the greedy caching algorithm is very close to the optimal one. Finally, the above observations do not change for different parameters of the file request probabilities. Note a large value of $\gamma_r$ means that the requests from mobile users are more concentrated on the popular files.

Fig. \ref{fig_scheme2} shows the effect of the caching capacity $C$. It is demonstrated that the proposed mobility-aware strategy always provides a substantial performance gain over the random caching and popular caching strategies for different caching capacities. In addition, when the caching capacity is very small, caching the segments of the most popular files provides a good performance.

Fig. \ref{fig_runtime1} shows the runtime of the proposed algorithms on a desktop computer with a $3.20$ GHz Intel Core i5 processor and $8$ GB installed memory. Even though the runtime of the optimal algorithm increases exponentially with the number of users, it can be easily applied to a small size network within one minute. Moreover, the runtime of the greedy algorithm is much lower than the optimal algorithm, and it can achieve a near optimal performance as  shown in Fig. \ref{fig_scheme1} and \ref{fig_scheme2}.



\subsection{Mobility Pattern}

In this part, we will investigate the effect of user mobility pattern on the proposed caching strategy. First, as the inter-contact model adopted in this study cannot directly capture the speed of user mobility, we will connect it with the random waypoint model \cite{rwp}, which can directly reflect user velocity. We will then investigate how user mobility affects the caching performance.

\begin{figure}
\begin{minipage}[t]{0.49\linewidth}
\centering
\includegraphics[width=2.2in]{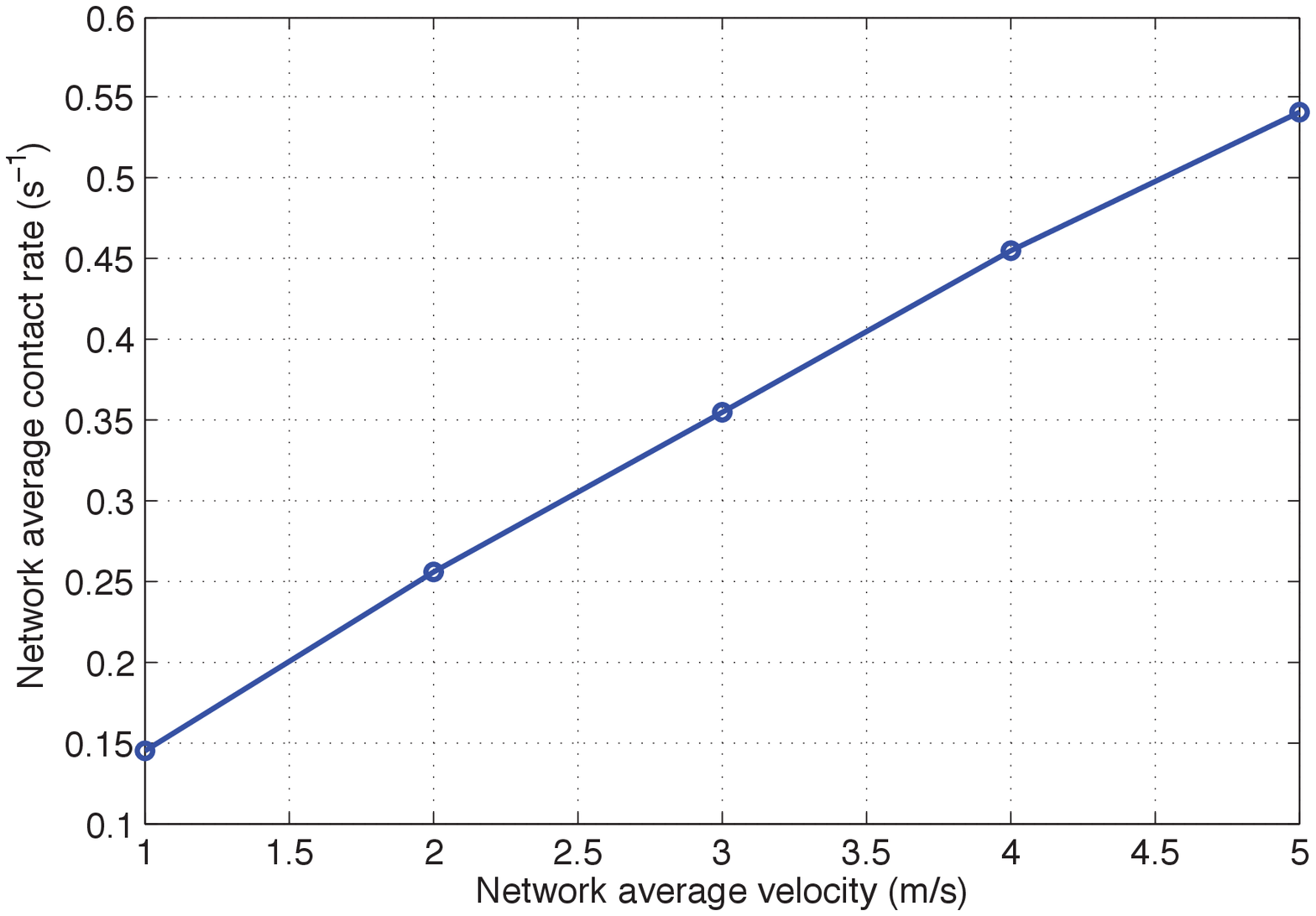}
  \caption{Network average contact rate.}
\label{fig_networkrate}
\end{minipage}%
\hspace{0.1in}
\begin{minipage}[t]{0.49\linewidth}
\centering
  \includegraphics[width=2.2in]{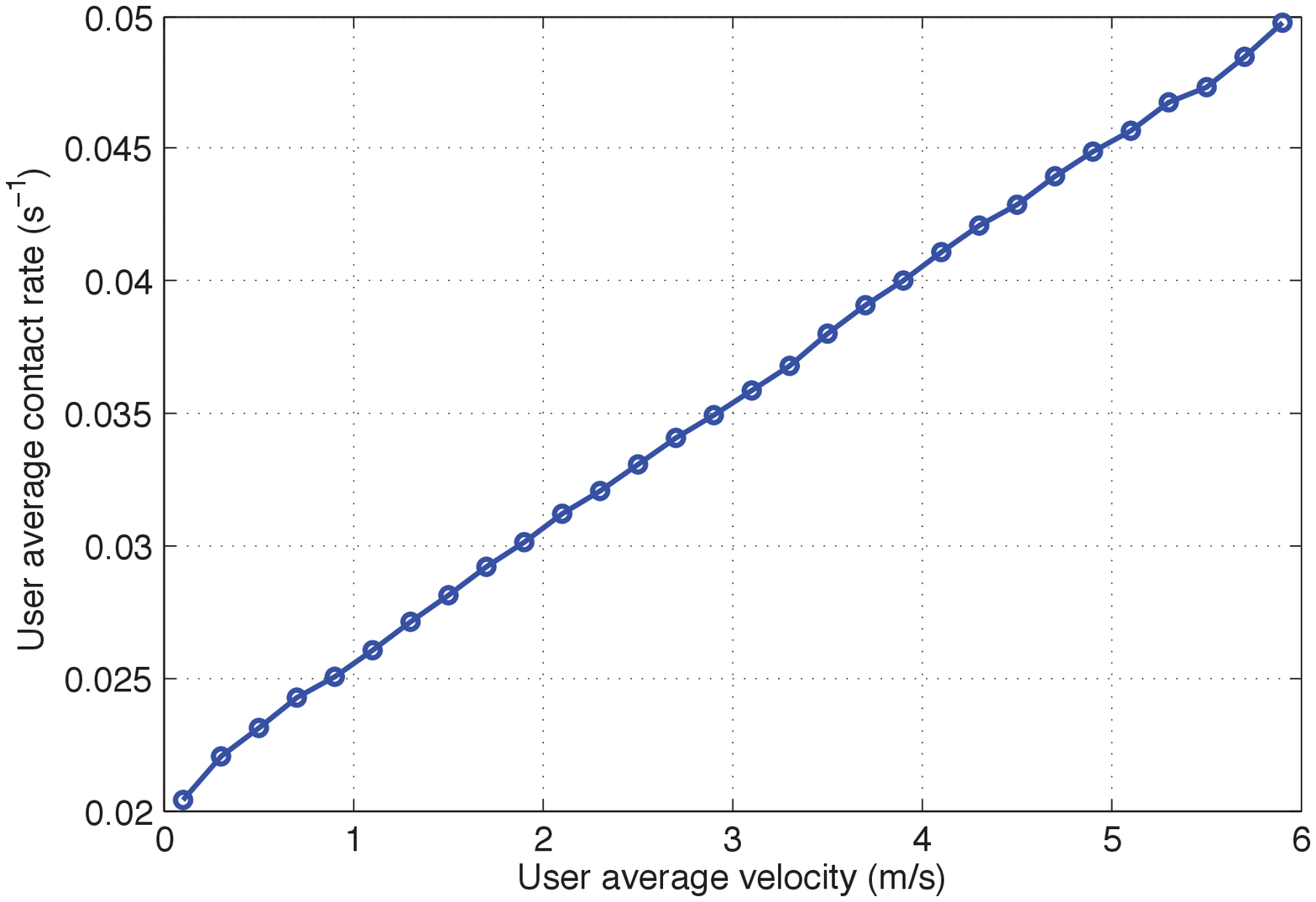}
  \caption{\scriptsize {User average contact rate.}}
\label{fig_userrate}
\end{minipage}
\end{figure}

In the random waypoint model, user $i$ randomly selects a target point in the area and moves towards it with velocity randomly selected in $(0,2 V_i^{a})$. When a user reaches its target, it will randomly select a new target. The user average velocity $V_i^{a}$ is uniformly distributed in $(0, 2 V^n)$, where $V^n$ is the network average velocity. Simulation results of user velocity and contact rates are shown in Fig. \ref{fig_networkrate} and Fig. \ref{fig_userrate}, where $20$ users are moving in a $300$ m $\times$ $300$ m area and the transmission range of each user is $30$ m. Fig. \ref{fig_networkrate} shows how the network average velocity affects the \emph{network average contact rate}, i.e., the average number of contacts in the network per unit time. A linear relationship is observed, i.e., the faster the users move, the more contacts will occur. Fig. \ref{fig_userrate} shows the relationship of the user average velocity and the \emph{user average contact rate}, i.e., the average number of contacts for a particular user per unit time. For this simulation, we use $V^n=3 m/s$. Again, it reveals a linear relationship between the user average velocity and the user average contact rate. Therefore, in the following investigation, the effect of contact rates can be regarded as to be equivalent to the effect of user velocity.

Next, we will show how the user mobility pattern affects the caching performance. The contact rate for each device pair, i.e., $\lambda_{i,j}$, for $i \in \mathcal{D}$ and $j \in \mathcal{D}\backslash \{i\}$, is generated according to the Gamma distribution as $\Gamma(k,\theta)$, which means that the average contact rate for the whole network is $k \theta$ and the variance is $k \theta ^2$. In Fig \ref{fig_scheme4}, we consider the effect of the network average contact rate, while the variance remains the same as $\Gamma(4.43, 1/1088)$ \cite{aggregate_real}. A larger network average contact rate implies a higher frequency of contacts in the network, or equivalently, a higher user velocity. We see that the data offloading ratio increases first linearly, then logarithmically  with the increase of the network average contact rate, with either mobility-aware caching or random caching. When the contact rate is large enough, each user may access the caching storage of all the users in the network, and the data offloading ratio reaches a fixed value.
The data offloading ratio using the popular caching increases slightly, and quickly reaches a fixed value when the network average contact rate increases. For both the random caching strategy and greedy mobility-aware strategy, mobile users cache different contents and utilize the D2D links to exchange data. Thus, both achieve a significant performance gain compared to the popular caching strategy. In addition to utilizing the D2D links, the greedy mobility-aware caching strategy also takes advantage of the user mobility pattern, and thus, it outperforms the random caching strategy.

\begin{figure}
\begin{minipage}[t]{0.49\linewidth}
\centering
  \includegraphics[width=2.2in]{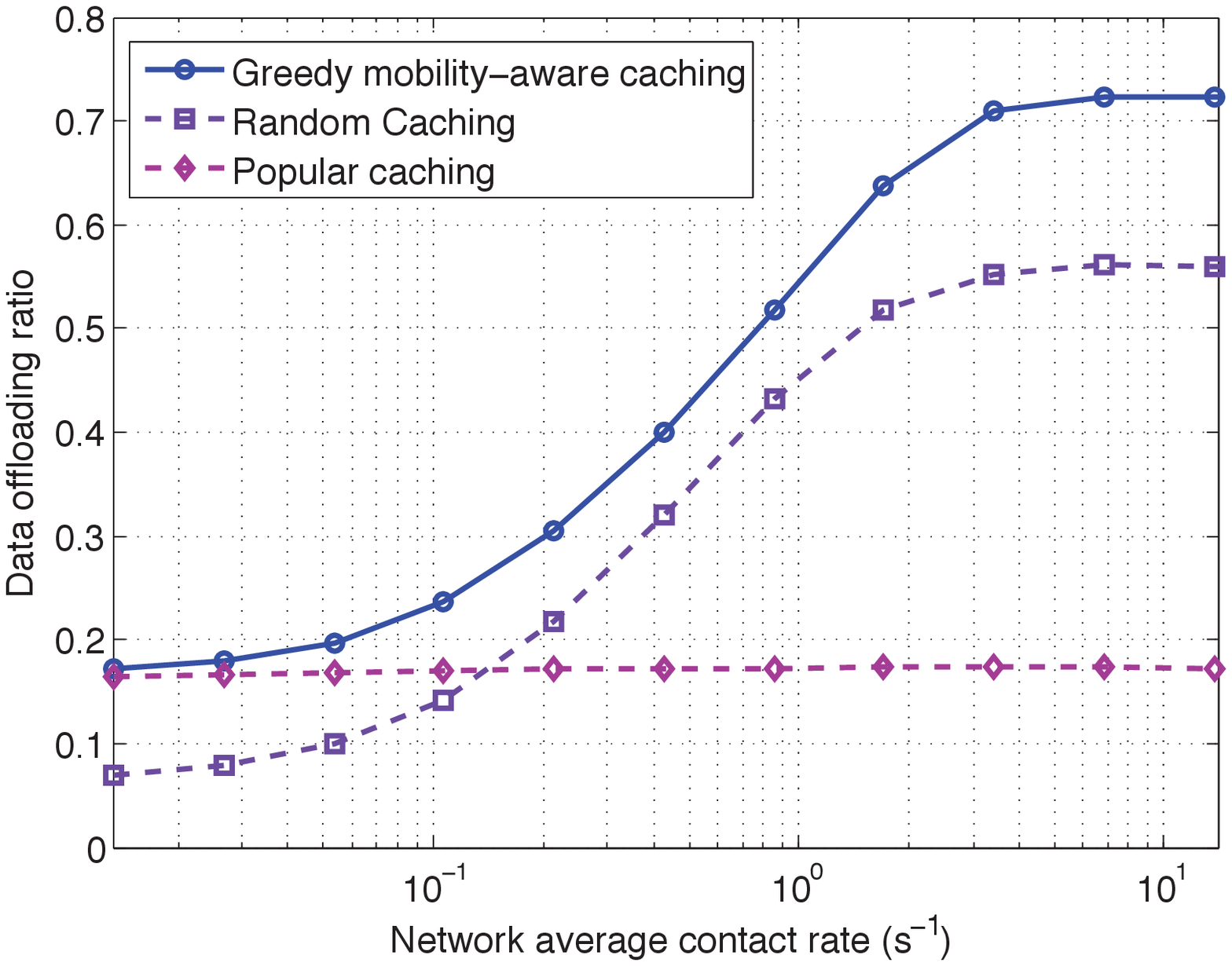}
  \caption{Comparison of different caching strategies with $N_u=15$, $N_{file}=100$, $T^d=120 \text{ s}$, $B_{i,j}=1$, $K_{\max}=3$ and $C=8$.}
\label{fig_scheme4}
\end{minipage}%
\hspace{0.1in}
\begin{minipage}[t]{0.49\linewidth}
\centering
  \includegraphics[width=2.25in]{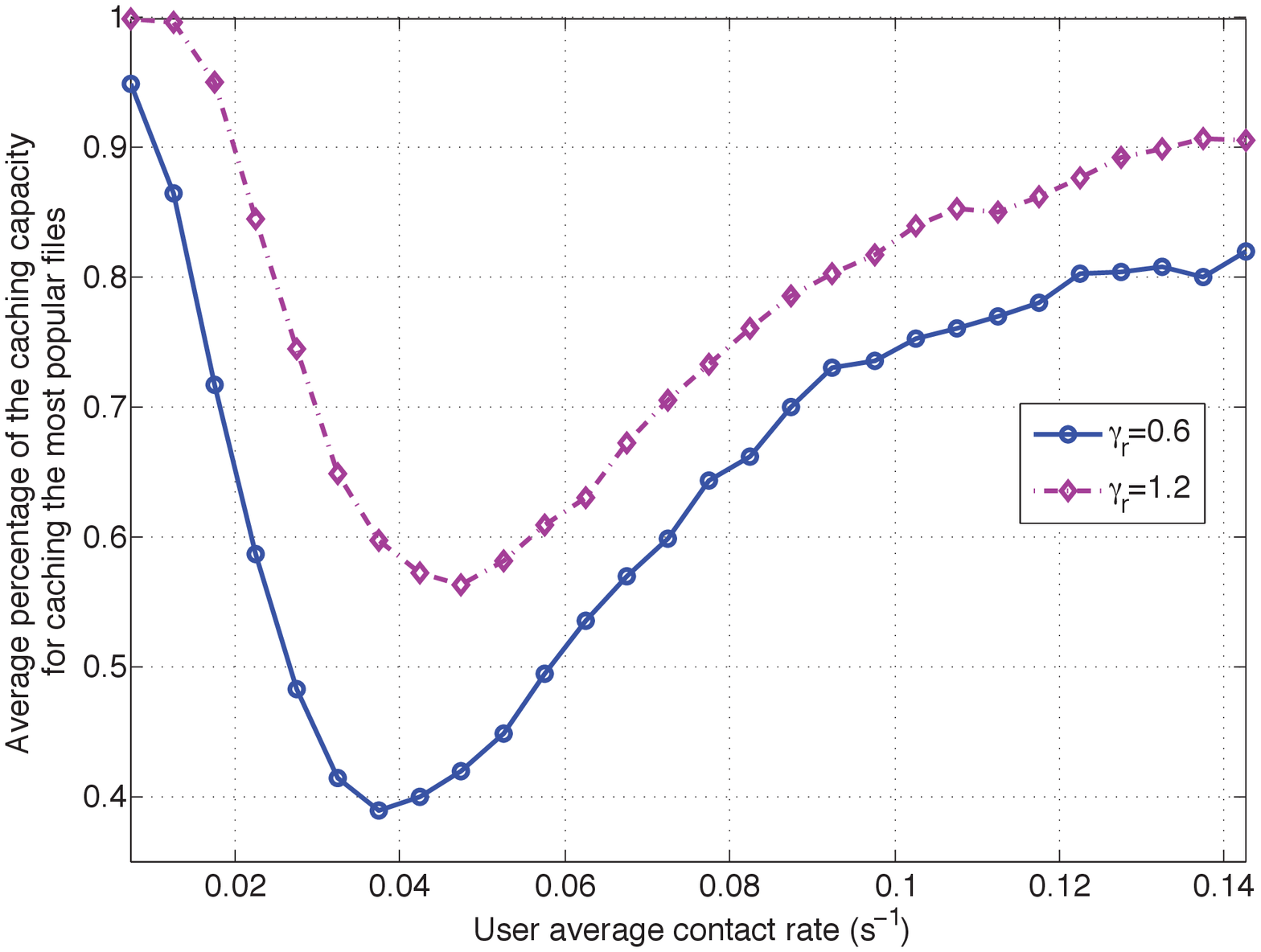}
  \caption{Average percentage of the caching capacity for caching the most popular files when $N_u=15$, $N_{file}=100$, $T^d=120 \text{ s}$, $B_{i,j}=1$, $K_{\max}=3$ and $C=8$. }
\label{fig_scheme5}
\end{minipage}
\end{figure}

Fig. \ref{fig_scheme5} illustrates how the caching behavior changes with the user average contact rate, which is the rate each user contacts all the other users and is related to the user velocity as shown in Fig. \ref{fig_userrate}. In particular, we will focus on the $C$ most popular files, and investigate the percentage of the caching capacity of each user that is used to cache these most popular files. In this simulation, we first generate the contact rate for each device pair, i.e., $\lambda_{i,j}$, for $i \in \mathcal{D}$ and $j \in \mathcal{D} \backslash \{i\}$, according to the Gamma distribution $\Gamma(0.443,1/108.8)$. Then, the user average contact rate of user $i$ can be calculated as $\lambda_i=\sum\limits_{j\in D\backslash \{i\}}\lambda_{i,j}$. For each realization, caching placement is obtained by applying the greedy mobility-aware caching strategy. Afterwards, for the $i$-th user, the percentage of its caching capacity used to cache the $C$ most popular files, denoted as $\eta_i$, will be recorded. In this way, we get one data point $(\lambda_i, \eta_i)$ for user $i$ in each realization. In total, $5000$ realizations are run. We discretize the user average contact rates, and plot the corresponding average percentage of the caching capacity for caching the most popular files in Fig. \ref{fig_scheme5}. We see that this percentage first decreases and then increases with the user average contact rate going large. Intuitively, when the user contact rate is low, users prefer to cache the most popular files in order to support their own requirements. When the user contact rate is relatively high, users are likely to cache the most popular files as well, but to help other users to download the requested files via D2D links. On the other hand, the users with medium user contact rates need to cache more diversified contents to avoid duplicate caching, as well as to exploit the user contact opportunities and D2D communications.

\subsection{Realistic Data-Sets}

\begin{figure}
\begin{minipage}[t]{0.31\linewidth}
\centering
\includegraphics[width=2in]{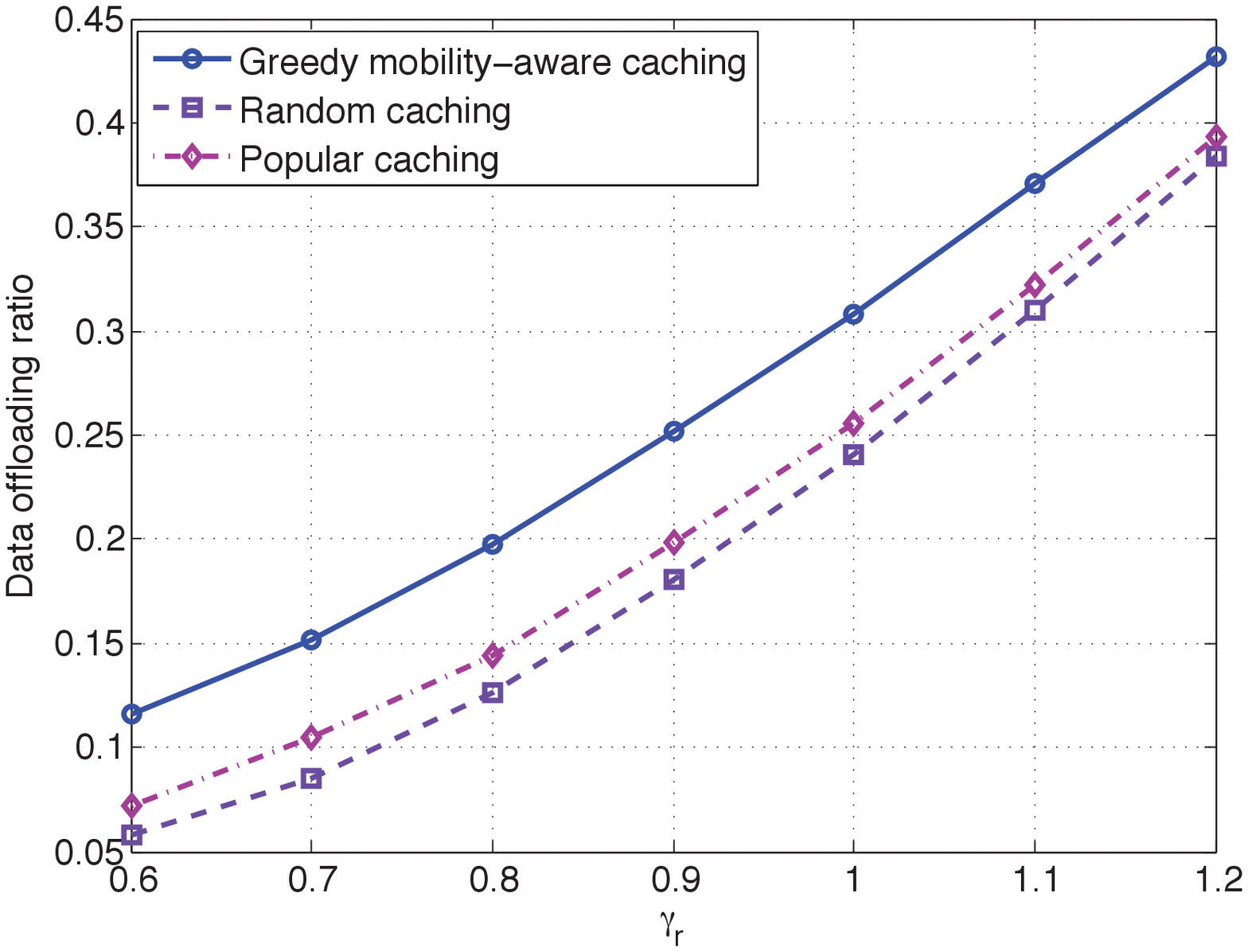}
  \caption{Comparison of different caching strategies with $N_{file}=500$, $T^d=120 \text{ s}$, $B_{i,j}=1$, $K_{\max}=5$ and $C=10$.}
\label{fig_scheme6}
\end{minipage}%
\hspace{0.1in}
\begin{minipage}[t]{0.31\linewidth}
\centering
  \includegraphics[width=1.95in]{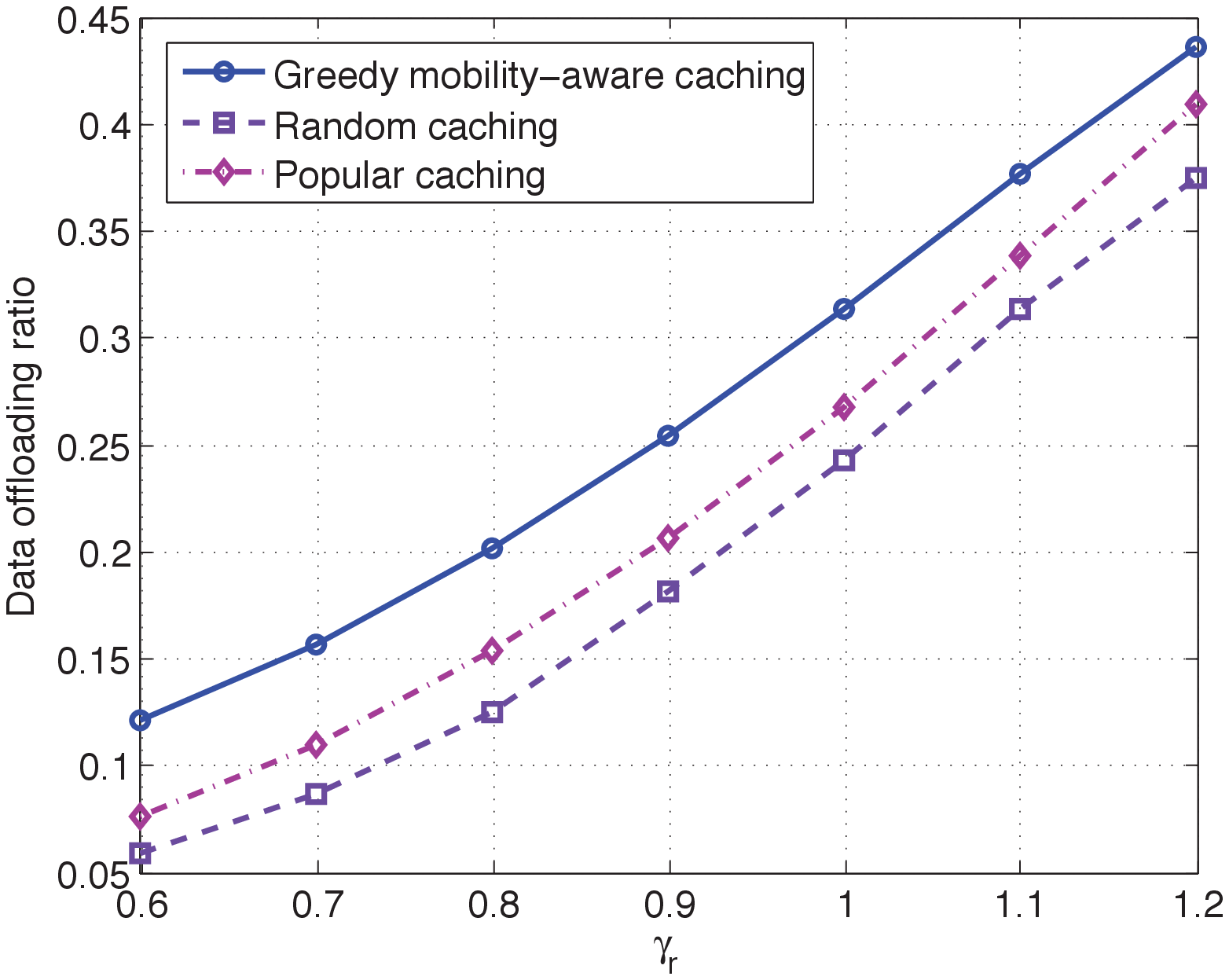}
  \caption{Comparison of different caching strategies with $N_{file}=500$, $T^d=600 \text{ s}$, $B_{i,j}=1$, $K_{\max}=5$ and $C=10$.}
\label{fig_scheme7}
\end{minipage}%
\hspace{0.1in}
\begin{minipage}[t]{0.31\linewidth}
\centering
  \includegraphics[width=2.06in]{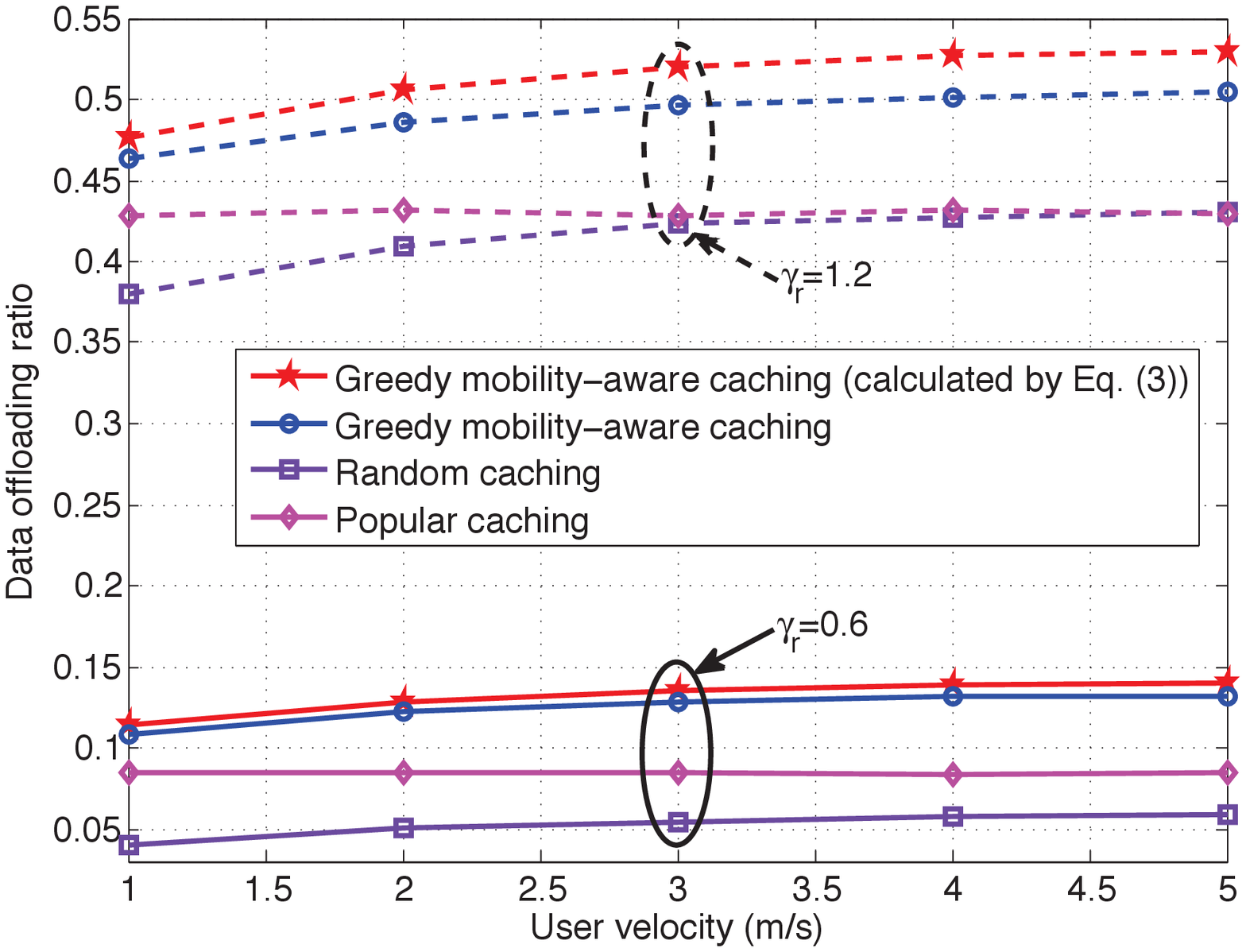}
		\caption{Comparison of different caching strategies with $N_{file}=500$, $T^d=120 \text{ s}$, $B_{i,j}=1$, $K_{\max}=3$ and $C=8$.}
\label{fig_SLAW}
\end{minipage}
\end{figure}

To verify the effectiveness of the proposed mobility-aware caching strategy, we utilize the data set conducted by Chaintreau \emph{et al.} in the INFOCOM 2006 conference \cite{CHSGCD07Impact}. IMotes, i.e., small portable Bluetooth radio devices, were distributed to 78 student workshop participants. The transmission range of each iMote is approximately 30 meters. It is found that user mobility patterns are different in the daytime and nighttime, i.e., users are much more active in the daytime, which fits our intuition. As an example, we use the daytime data during the first day and estimate the contact rate for each device pair, i.e., $\lambda_{i,j}$, by the average number of contact times per second\footnote{Notice that this is a very simple estimator, but it already provides a substantial performance gain with mobility-aware caching. With more data and more sophisticated estimators, the performance can be further improved.}. Based on the estimation, we then apply three different caching placement strategies, i.e., the greedy mobility-aware caching strategy, the random caching strategy and the popular caching strategy. Fig. \ref{fig_scheme6} compares the performance of the three caching strategies during the daytime in the second day. It is shown that the greedy caching strategy outperforms the random caching strategy by $12\% \sim 100\%$ and outperforms the popular caching strategy by $10\% \sim 60\%$, which validates the user mobility model and demonstrates the advantage of the proposed mobility-aware caching strategy.

Besides the conference scenario, we also test the proposed mobility-aware caching strategy in the campus environment, where iMotes are distributed to 36 students in the Cambridge campus \cite{cam-leguay06}. Based on the data during the daytime in previous days, we estimate the contact rate for each device pair, i.e., $\lambda_{i,j}$, and design the caching strategies. The performance of the three caching strategies in the following day is shown in Fig. \ref{fig_scheme7}. We see that the greedy caching strategy outperforms the random caching strategy by $15\% \sim 100\%$ and outperforms the popular caching strategy by $7\% \sim 60\%$.

To show the effect of mobility in a realistic scenario, we generate the user moving traces based on the Self-similar Least-Action Walk (SLAW) model \cite{SLAW}, where $80$ users are moving with the same velocity in a $1000$ m $\times$ $1000$ m area and the transmission range of each user is $30$ m. In the simulation, the hurst parameter for self-similarity of waypoints is set as $0.8$, and the pause-time follows a truncated Pareto distribution with the minimum and maximum values as $30$ s and $30$ min, respectively. We first estimate the contact rate for each device pair, i.e., $\lambda_{i,j}$, and design the caching strategies. The performance of the three strategies on the generated realistic user traces is shown in Fig. \ref{fig_SLAW}. Besides, the numerical value calculated by Eq. (\ref{problem}) is shown to demonstrate that the Poission model provides a good approximation. It also illustrates that the data offloading ratio increases with the user velocity, with either greedy mobility-aware caching or random caching. Meanwhile, the greedy mobility-aware caching strategy outperforms both random and popular caching strategies with different user velocities.

\section{conclusions}
In this paper, we exploited user mobility to improve caching placement in D2D networks using a coded cache protocol. We took advantage of the inter-contact pattern of user mobility when formulating the caching placement problem. To assist the evaluation of the complicated objective function, we proposed a divide and conquer algorithm. A DP algorithm was then developed to find the optimal caching placement, which is much more efficient than exhaustive search. By reformulating it as a monotone submodular maximization problem over a matroid constraint, we developed an effective greedy caching placement algorithm, which achieves a near-optimal performance. Simulation results based on both the mathematical model and real-life data-sets showed that the proposed mobility-aware greedy caching strategy outperforms both the random caching strategy and popular caching strategy. It is observed that slow users tend to cache the most popular files to support themselves, since they have few opportunities to communicate with others via D2D links. Fast moving users also tend to cache the most popular files, but the purpose is mainly to fulfill the requests from other users, since they have much more opportunities to establish D2D links with others. Meanwhile, users with medium velocity tend to cache less popular files to better exploit the D2D communications. For future works, it would be interesting to exploit other information of the user mobility pattern to further improve the caching efficiency. It will also be interesting to investigate joint caching at both BSs and devices, as well as distributed algorithms for implementation in large-size networks.

\section*{Appendix}
\subsection{Proof of Lemma \ref{NP}}

We will prove that problem (\ref{problem}) is NP-hard by a reduction from the $2$-disjoint set covers (2DSC) problem, which has been proved to be NP-hard \cite{cardei2005improving}. 
\begin{defn}
2DSC: Denote $\mathcal{B}$ as a set of finite elements, and $\mathcal{A}$ as a collection of subsets of $\mathcal{B}$. The problem is to determine whether there exist two subsets of $\mathcal{A}$, i.e., $\mathcal{A}_1 \subset \mathcal{A}$ and $\mathcal{A}_2 \subset \mathcal{A}$, such that: 1) $\mathcal{A}_1 \cap \mathcal{A}_2 = \varnothing$; 2) The union of the members in $A_i, i=1,2$ covers all the elements in $\mathcal{B}$, i.e., $\mathcal{B}=\bigcup_{\mathcal{B}' \in \mathcal{A}_i} \mathcal{B}', i=1,2.$
\end{defn}
The 2DSC problem can be solved by solving a special case of problem (\ref{problem}), where $N_f=2$, $K_f=1$, $C=1$, $B_{i,j}=1$ $(i \in \mathcal{D}, j \in \mathcal{D})$, and $N_u=\max(|\mathcal{A}|,|\mathcal{B}|)$. Denote the $i$-th element in $\mathcal{B}$ as $b_j$ and the $j$-th member in $\mathcal{A}$ as $\mathcal{A}^m_j$, and if $b_i \in \mathcal{A}^m_j$, set $\lambda_{i,j}=+\infty$; otherwise, set $\lambda_{i,j}=0$. Meanwhile, set $\lambda_{i,j}=0$, if $i>|\mathcal{B}|$ or $j>|\mathcal{A}|$. Then, the special case of problem (\ref{problem}) becomes

\begin{align}
\mathop {\max }\limits_{X } \quad & \frac{1}{N_u} \sum\limits_{i \in \mathcal{D}} \sum\limits_{f \in \{1,2\}} p_f \left[1-\exp \left(-T^d \sum \limits_{j \in \mathcal{D}}\lambda_{i,j}x_{j,f} \right)\right] \notag \\
	&=\frac{1}{N_u} \sum\limits_{i \in \mathcal{B}} \sum\limits_{f \in \{1,2\}} p_f \mathbbm{1}_{\exists j \in \mathcal{A}, \text{ s.t. } b_i \in \mathcal{A}^m_j \text{ and } x_{j,f}=1} , \label{problem_sim} \\
\text{s.t.} \quad & x_{j,1}+x_{j,2} \leqslant 1 , \forall j \in \mathcal{A}, \tag{\ref{problem_sim}a} 
\end{align}
\begin{align}
& x_{j,f} \in \{0,1\}, \forall j \in \mathcal{A} \text{ and } f \in \{1,2\}, \tag{\ref{problem_sim}b} 
\end{align}
Accordingly, if the solution of problem (\ref{problem_sim}) is $\frac{|\mathcal{B}|}{N_u}$, two disjoint subsets of $\mathcal{A}$ are found, i.e., $\mathcal{A}_1=\{\mathcal{A}^m_j| x_{j,1}=1\}$ and $\mathcal{A}_2=\{\mathcal{A}^m_j| x_{j,2}=1\}$; otherwise, there is no such subsets. Thus, if Problem (\ref{problem}) can be solved in polynomial time, the 2DSC problem can be solved in polynomial time, i.e., the 2DSC problem is poly-time reducible to Problem (\ref{problem}). Since an NP-hard problem is reducible to problem (\ref{problem}), problem (\ref{problem}) is NP-hard.

\subsection{Proof of Lemma 4}
The gain of adding $y_{j,f,k}$ into the caching placement $Y$ is
\begin{align}
&g(Y \cup \{y_{j,f,k}\})-g(Y)= \frac{p_f }{N_u K_f} \times \Bigg\{ \notag \\
&(a)
\begin{cases}
\begin{array}{c}
\sum \limits_{q=0} ^{K_f-1} q \Bigg(\sum \limits_{i \in \mathcal{D}} \Pr \bigg[ \sum_{j' \in \mathcal{D}, j' \ne j}
\min (B_{i,j'}M_{i,j'}, |Y \cap S_{j',f}|) + \min (B_{i,j} M_{i,j}, |Y \cap S_{j,f}|+1) = q \bigg]
\end{array} & \\
-\sum \limits_{i \in \mathcal{D}} \Pr \bigg[ \sum \limits_{j' \in \mathcal{D}, j' \ne j}
\min (B_{i,j'}M_{i,j'}, |Y \cap S_{j',f}|) + \min (B_{i,j} M_{i,j}, |Y \cap S_{j,f}|) = q \bigg] \Bigg) &
\end{cases} \notag \\
&(b)
\begin{cases}
\begin{array}{c}
+ K_f \Bigg( \sum \limits_{i \in \mathcal{D}} \Pr \bigg[ \sum_{j' \in \mathcal{D},j' \ne j}
\min (B_{i,j'}M_{i,j'}, |Y \cap S_{j',f}|) + \min (B_{i,j} M_{i,j}, |Y \cap S_{j,f}|+1) \ge K_f \bigg]
\end{array} & \\
-\sum \limits_{i \in \mathcal{D}} \Pr \bigg[ \sum \limits_{j' \in \mathcal{D}, j' \ne j}
\min (B_{i,j'}M_{i,j'}, |Y \cap S_{j',f}|) + \min (B_{i,j} M_{i,j}, |Y \cap S_{j,f}|) \ge K_f \bigg] \Bigg) \Bigg\} &
\end{cases} \label{gain_origin}
\end{align}
We first divide the Eq. (\ref{gain_origin}) into two parts, and tackle them separately. The first part is
\begin{align}
(a)
= & \sum \limits_{q=0} ^{K_f-1} q \Bigg(\sum \limits_{i \in \mathcal{D}}  \sum \limits_{q'=|Y \cap S_{j,f}|}^{|Y \cap S_{j,f}|+1} \Pr \bigg[ \sum_{j' \in \mathcal{D}, j' \ne j}
\min (B_{i,j'}M_{i,j'}, |Y \cap S_{j',f}|) = q-q' \bigg] \notag \\
& \times \left( \Pr \left[ \min (B_{i,j} M_{i,j}, |Y \cap S_{j,f}|+1) = q' \right]
- \Pr \left[ \min (B_{i,j} M_{i,j}, |Y \cap S_{j,f}|) = q' \right] \right) \Bigg).
\end{align}
When $ B_{i,j} \left( \left\lceil \frac{|Y \cap S_{j,k}|+1}{B_{i,j}}\right\rceil -1 \right)=|Y \cap S_{j,f}|$,
\begin{align}
(a)= & \sum \limits_{q=0} ^{K_f-1} q \sum \limits_{i \in \mathcal{D}} \Bigg( \Pr \bigg[ \sum_{j' \in \mathcal{D}, j' \ne j}
\min (B_{i,j'}M_{i,j'}, |Y \cap S_{j',f}|) = q-|Y \cap S_{j,f}| \bigg] \notag \\
& \times \left( \Pr \left[ M_{i,j} = \left\lceil \frac{|A \cap S_{j,k}|+1}{B_{i,j}}\right\rceil-1 \right] - \Pr \left[ M_{i,j} \ge \left\lceil \frac{|A \cap S_{j,k}|+1}{B_{i,j}}\right\rceil-1 \right] \right) \notag 
\end{align}
\begin{align}
& + \Pr \bigg[ \sum_{j' \in \mathcal{D}, j' \ne j}
\min (BM_{i,j'}, |Y \cap S_{j',f}|) = q-|Y \cap S_{j,f}|-1 \bigg] \notag \\
& \times \Pr \left[ M_{i,j} \ge \left\lceil \frac{|A \cap S_{j,k}|+1}{B_{i,j}}\right\rceil \right]  \Bigg) \notag \\
= & \sum \limits_{q=0} ^{K_f-1} q \sum \limits_{i \in \mathcal{D}}
 \Pr \left[ M_{i,j} \ge \left\lceil \frac{|A \cap S_{j,k}|+1}{B_{i,j}}\right\rceil \right]  \notag \\
 & \times \Bigg( \Pr \bigg[ \sum_{j' \in \mathcal{D}, j' \ne j}
\min (B_{i,j'}M_{i,j'}, |Y \cap S_{j',f}|) = q-|Y \cap S_{j,f}|-1 \bigg] \notag \\
&-\Pr \bigg[ \sum_{j' \in \mathcal{D}, j' \ne j}
\min (B_{i,j'}M_{i,j'}, |Y \cap S_{j',f}|) = q-|Y \cap S_{j,f}| \bigg] \Bigg).
\end{align}
When $ B_{i,j} \left( \left\lceil \frac{|Y \cap S_{j,k}|+1}{B_{i,j}}\right\rceil -1 \right)<|Y \cap S_{j,f}|$,
\begin{align}
(a)= & \sum \limits_{q=0} ^{K_f-1} q \sum \limits_{i \in \mathcal{D}} \Bigg( \Pr \bigg[ \sum_{j' \in \mathcal{D}, j' \ne j}
\min (B_{i,j'}M_{i,j'}, |Y \cap S_{j',f}|) = q-|Y \cap S_{j,f}| \bigg] \notag \\
& \times \left( 0 - \Pr \left[ M_{i,j} \ge \left\lceil \frac{|A \cap S_{j,k}|+1}{B_{i,j}}\right\rceil \right] \right) \notag \\
& + \Pr \bigg[ \sum_{j' \in \mathcal{D}, j' \ne j}
\min (B_{i,j'}M_{i,j'}, |Y \cap S_{j',f}|)  = q-|Y \cap S_{j,f}|-1 \bigg] \notag \\
& \times \Pr \left[ M_{i,j} \ge \left\lceil \frac{|A \cap S_{j,k}|+1}{B_{i,j}}\right\rceil \right]  \Bigg) \notag
\end{align}
\begin{align}
= & \sum \limits_{q=0} ^{K_f-1} q \sum \limits_{i \in \mathcal{D}}
 \Pr \left[ M_{i,j} \ge \left\lceil \frac{|A \cap S_{j,k}|+1}{B_{i,j}}\right\rceil \right]  \notag \\
 & \times \Bigg( \Pr \bigg[ \sum_{j' \in \mathcal{D}, j' \ne j}
\min (B_{i,j'}M_{i,j'}, |Y \cap S_{j',f}|) = q-|Y \cap S_{j,f}|-1 \bigg] \notag \\
&-\Pr \bigg[ \sum_{j' \in \mathcal{D}, j' \ne j}
\min (B_{i,j'}M_{i,j'}, |Y \cap S_{j',f}|)  = q-|Y \cap S_{j,f}| \bigg] \Bigg).
\end{align}
Thus, we can get
\begin{align}
(a)= & \sum \limits_{i \in \mathcal{D}}
 \Pr \left[ M_{i,j} \ge \left\lceil \frac{|A \cap S_{j,k}|+1}{B_{i,j}}\right\rceil \right]  \notag \\
& \times \Bigg( \Pr \bigg[ \sum_{j' \in \mathcal{D}, j' \ne j}
\min (B_{i,j'}M_{i,j'}, |Y \cap S_{j',f}|)\le K_f-|Y \cap S_{j,f}|-1 \bigg] \notag 
\end{align}
\begin{align}
&-K_f \Pr \bigg[ \sum_{j' \in \mathcal{D}, j' \ne j}
\min (B_{i,j'}M_{i,j'}, |Y \cap S_{j',f}|)= K_f-|Y \cap S_{j,f}|-1 \bigg] \Bigg).
\end{align}
Similarly, we can get
\begin{align}
(b)
= & \sum \limits_{i \in \mathcal{D}}
 \Pr \left[ M_{i,j} \ge \left\lceil \frac{|A \cap S_{j,k}|+1}{B_{i,j}}\right\rceil \right]  \notag \\
 & \times K_f \Pr \bigg[ \sum_{j' \in \mathcal{D}, j' \ne j}
\min (B_{i,j'}M_{i,j'}, |Y \cap S_{j',f}|) = K_f-|Y \cap S_{j,f}|-1 \bigg],
\end{align}
By combining the four parts, the gain of adding $y_{j,f,k}$ into the caching placement $Y$ can be given as shown in Lemma 3.

\subsection{Proof of Lemma 5}

To prove that the function $g(A)$ defined in (\ref{objective}) is a monotone submodular function, we will show that $g(A)$ satisfies Proposition 1. Let $ A \subset D \subset S$, $E=D-A$ and $y_{j,f,k} \in S-D$, then according to Eq. (\ref{gain}), we can get
\begin{align}
g(A \cup \{y_{j,f,k}\})&-g(A)=\frac{p_f }{N_u K_f}  \sum \limits_{i \in \mathcal{D}} \Pr \bigg[ M_{i,j} \geqslant \left\lceil \frac{|D \cap S_{j,f}-E|+1}{B_{i,j}}\right\rceil \bigg] \notag \\
& \times \Pr \bigg[\sum_{j' \in \mathcal{D}, j' \ne j} \min (B_{i,j'}M_{i,j'}, |D \cap S_{j',f}|-E)   \le K_f- |D \cap S_{j,f}-E|-1 \bigg]
\end{align}
Since $\min (B_{i,j'} M_{i,j'}, |D \cap S_{j',f}-E|) \leqslant \min (B_{i,j'} M_{i,j'}, |D \cap S_{j',f}|)$, denote $U_{i,j'}= \min (B_{i,j'} M_{i,j'}, |D \cap S_{j',f}|) - \min (B_{i,j'} M_{i,j'}, |D \cap S_{j',f}-E|) $, then $U_{i,j'} \geqslant 0$ and
\begin{align}
g(A & \cup \{y_{j,f,k}\})-g(A) \notag \\
=&
\frac{p_f }{N_u K_f}  \sum \limits_{i \in \mathcal{D}} \Pr \bigg[ M_{i,j} \geqslant \left\lceil \frac{|D \cap S_{j,f}|-|E \cap S_{j,f}|+1}{B_{i,j}}\right\rceil \bigg] \notag \\
& \times \Pr \bigg[\sum_{j' \in \mathcal{D}, j' \ne j} \min (B_{i,j'}M_{i,j'}, |D \cap S_{j',f}|)  \le K_f- |D \cap S_{j,f}|-1 +|E \cap S_{j,f}|+\sum_{j' \in \mathcal{D}, j' \ne i,j} U_{i,j'} \bigg] \notag 
\end{align}
\begin{align}
\geqslant&
\frac{p_f }{N_u K_f}  \sum \limits_{i \in \mathcal{D}} \Pr \bigg[ M_{i,j} \geqslant \left\lceil \frac{|D \cap S_{j,f}|+1}{B_{i,j}}\right\rceil \bigg] \notag \\
& \times \Pr \bigg[\sum_{j' \in \mathcal{D}, j' \ne j} \min (B_{i,j'}M_{i,j'}, |D \cap S_{j',f}|)  \le K_f- |D \cap S_{j,f}|-1 \bigg] \notag \\
= & g(D \cup \{y_{j,f,k}\})-g(D).
\end{align}
It is easy to get the conclusion that $g(A \cap \{y_{j,f,k}\})-g(A) \geqslant 0 $, since it is a summation of non-negative probabilities. In conclusion, the function $g(A)$ satisfies Proposition 1, and thus, it is a monotone submodular function.

\bibliographystyle{IEEEtran}
\bibliography{IEEEabrv,report}
\end{document}